\newcounter{myctr}
\def\myitem{\refstepcounter{myctr}\bibfont\noindent\ifnum\themyctr>9\else\phantom{0}\fi\hangindent17pt\themyctr.\enskip}

\documentclass{ws-ijqi}
\usepackage{hyperref}
\usepackage[super,sort,compress]{cite}

\usepackage{array}
\usepackage{amsmath}
\usepackage{amssymb}
\usepackage{graphicx}
\usepackage{esint}
\usepackage{xcolor}
\usepackage{MnSymbol}
\usepackage{hyperref}
\usepackage[english]{babel}
\usepackage[latin9]{inputenc}

\begin{document}

\catchline{}{}{}{}{}

\title{ANALYSIS OF SINGLE-PARTICLE NONLOCALITY THROUGH THE PRISM OF WEAK MEASUREMENTS}

\author{DANKO GEORGIEV}

\address{Institute for Advanced Study, 30 Vasilaki Papadopulu Str. \\ Varna 9010, Bulgaria \\ danko.georgiev@mail.bg}

\author{ELIAHU COHEN}

\address{Faculty of Engineering and the Institute of Nanotechnology and Advanced Materials,
Bar Ilan University \\ Ramat Gan 5290002, Israel\\
eliahu.cohen@biu.ac.il}

\maketitle

\begin{history}
\received{15/11/2019}
\end{history}

\begin{abstract}
Although regarded today as an important resource in quantum information, nonlocality has yielded over the years many conceptual conundrums. Among the latter are nonlocal aspects of single particles which have been of major interest. In this paper, the nonlocality of single quanta is studied in a square nested Mach--Zehnder interferometer with spatially separated detectors using a delayed choice modification of quantum measurement outcomes that depend on the complex-valued weak values. We show that if spacelike separated Bob and Alice are allowed to freely control their quantum devices, the geometry of the setup constrains the local hidden variables models. In particular, hidden signaling and a list of contextual instructions are required to split a quantum state characterized by a positive Wigner function into two quantum states with non-positive Wigner functions. This implies that local hidden variables models could rely neither on only two hidden variables for position and momentum, nor on simultaneous factorizability of both the hidden probability densities and weights of splitting to reproduce the correct quantum distributions. While our analysis does not fully exclude the existence of nonfactorizable local hidden variables models, it demonstrates that the recently proposed weak values of quantum histories necessitate contextual splitting of prior commitments to measurement outcomes, due to functional dependence on the total Feynman sum that yields the complex-valued quantum probability amplitude for the studied quantum transition. This analysis also highlights the quantum nature of weak measurements.
\end{abstract}



\markboth{D. Georgiev and E. Cohen}{Analysis of single-particle nonlocality through the prism of weak measurements}

\section{Introduction}

In 1927 at the \emph{Fifth Solvay International Conference on Electrons
and Photons}, Albert Einstein presented a thought experiment, which
conceptualized the nonlocality embedded in the collapse of the wavefunction of single quanta at the time of measurement \cite{Einstein1927}.
Einstein's experiment could be constructed with the use of a single-photon source, a beam splitter and two detectors, which are never found to click together at the same time as required by the conservation
of energy \cite{Guerreiro2012,Suarez2013}.
Since action at a distance contradicts the spirit of relativity theory, Einstein argued that quantum mechanics should be completed by the addition of local hidden variables that describe individual quantum processes \cite{Einstein1927}.
According to such viewpoint, individual quanta could preserve
local realism by always taking a single
path on their route to one of the measuring devices provided that these paths
remain hidden to direct experimental observation \cite{Selleri1988}.
Yet, subsequent research of entangled quantum systems \cite{Bell1964,Bell1966,Bertlmann2002,Aspect1982a,Aspect1982b,Hardy1992,Lundeen2009,Yokota2009} has shown that imposing locality on hidden variables leads to incorrect prediction of quantum results.

Classical realism demands that all physical observables have exact pre-determined values independently of whether they are measured \cite{Jerger2016}. The Kochen--Specker theorem, however, shows that quantum measurement outcomes are contextual and it is impossible for all quantum observables to have exact pre-determined values that are then revealed by the measurement \cite{Kochen1967,Mermin1990,Peres1990,Peres1991,Conway2006,Plastino2010,Leifer2014}.
Although classical systems are able to behave in a contextual manner \cite{Blasiak2015,Gondran2017,Li2017}, classical contextuality is based on a list of pre-determined input-output instructions that demand physical memory \cite{Kleinmann2011,Karanjai2018,Cabello2018}.
Einstein's revision of classical mechanics further asserts that classical systems should be relativistic and
integrate only local information that is transmitted at most at luminal speed \cite{Karimi2010,Karimi2015}.
In contrast, quantum systems
utilize nonlocal information about the settings of distant physical devices to enforce quantum correlations between distant measurement outcomes
at apparently superluminal speed \cite{Aharonov2000,Bancal2012,Gisin2014}.
Thus, exploiting the known link between the Kochen--Specker and Bell theorems \cite{Aolita2012,Genovese2019}, quantum contextuality of spatially separated entangled quantum systems can be manifested as a form of nonlocality.
Violation of Bell's inequalities \cite{Bell1964,Bell1966,Bertlmann2002} by entangled pairs of quanta has been extensively studied (for a detailed review see Refs.~\refcite{Genovese2019,Genovese2005})
and quantum nonlocality has been confirmed in numerous experiments \cite{Christensen2013,Giustina2015,Shalm2015,Abellan2018}.
The possible demonstration of quantum nonlocality with single quanta, however, has been debated and contested
\cite{Dunningham2007,Fuwa2015,Lee2017,Greenberger1995,Hardy1994,Hardy1995,Tan1991,Vaidman1995,Paraoanu2011}.

Previous works have studied the relations between contextuality and anomalous weak values \cite{Pusey2014,Piacentini2016Exp}. In this work we utilize a theorem \cite{Georgiev2018}, which relates the quantum probability amplitudes of individual virtual Feynman histories $\psi_i$ (defined by multi-time projection operators) and the sequential \cite{Mitchison2007,Piacentini2016Meas} weak values $A_w^{(i)}$ \cite{Aharonov1988,Dressel2012,Aharonov2014,Dressel2014,Dressel2015} of those Feynman histories, namely, the sequential weak value of an individual Feynman history $A_w^{(i')}$ is given by the ratio of the quantum probability amplitude for that particular history $\psi_{i'}$ and the sum of all coherently superposed Feynman histories with the same initial and final quantum states, $A_w^{(i')}=\frac{\psi_{i'}}{\sum_i \psi_{i}}$ (for further details see Section~\ref{sec:2}).
With the use of the latter theorem, we have designed an interferometric setup in which spacelike separated Alice and Bob are able to perform a delayed choice modification of the complex-valued weak values of alternative quantum histories \cite{Georgiev2018} in an attempt to understand better the nonlocal aspects of single quanta.
Proving that single quanta are nonlocal in the form of a ``no-go theorem'' requires mathematical demonstration that there exists no local hidden variables model that is capable to reproduce the standard quantum mechanical predictions. This is a challenging task because one needs to consider all possible strategies that local hidden variables models can attempt without violating Einstein's locality principle (for details see Section~\ref{sec:LHV}). Furthermore, by definition the local variables have ``hidden'' distributions that are unknown to us, hence there are very few mathematical operations that can be explicitly performed. This means that in order to derive a contradiction between the predictions of local hidden variables models and the predictions of standard quantum mechanics, one can use only very general mathematical constraints such as non-negative and normalized local hidden variable distributions, which are expected (but not guaranteed) to reproduce the correct quantum distributions for measured quantum observables after marginalization over the hidden variables.
Throughout the present theoretical study, we consider only predictions of the distributions that would be observed experimentally. We denote as ``quantum'' only the predictions obtained with the use of the standard quantum mechanical formalism. The predictions by the local hidden variables models are then ``required'' to reproduce the quantum outcomes after the marginalization over the hidden variables is performed. In other words, quantum nonlocality can be tested experimentally only if there are differences in the predictions by standard quantum mechanics and the most general local hidden variables model. Otherwise, if all quantum predictions can be replicated by a suitably constructed local hidden variables model, the setup cannot be claimed to demonstrate quantum nonlocality (for further details on this point see \ref{app:nonlocal}, \ref{app:Einstein} and \ref{app:Wheeler}).
In the interferometric setup that we propose, the quantum state before the last beam splitter is characterized by a positive Wigner function, whereas after the beam splitter the quantum evolution generates two quantum states with non-positive Wigner functions at each of the two detectors. This means that the generalized local hidden variables model cannot simply use the Wigner functions as hidden elements of reality because the hidden variable probability distributions have to be non-negative.
Because the Wigner function is the unique distribution in phase space that reproduces all rotated quadratures after marginalization, this also prevents the local hidden variables model from having only two hidden variables for position and momentum, necessitating introduction of extra hidden variables for measurement of different rotated quadratures.
The geometry of the setup also rules out local hidden variables models that rely on simultaneous factorizability of the hidden probability densities and the hidden weights of splitting to reproduce the correct quantum distributions. Deciding whether nonfactorizable local hidden variables models of the setup exist may require development of mathematical techniques for solving of systems of integral equations with unknown nonfactorizable kernels, which we leave for future work. In addition, the various tools employed here including the sum over histories approach, the Wigner function treatment and the derivations of various pointer distributions might be of interest for the weak value community.

\section{Square nested Mach--Zehnder interferometer\label{sec:2}}

In the setup shown in Fig.~\ref{fig:1},
Alice and Bob each have access to only one of two interferometer arms,
which are separated by a large distance.
On path~1, Alice measures an observable of the pointer that is a function of the weak value $(x_{1})_{w}$ of the position projector $\hat{x}_{1}=|x_{1}\rangle\langle x_{1}|$
using a measuring device~$M$.
On path~3, Bob chooses the value of a phase shifter $\varphi\in(-\pi,\pi)$.
In comparison to an earlier modification of Einstein's experiment \cite{Hardy1994}, the presented setup is a genuine single quantum experiment because the quantum under study is in an eigenstate of the particle number operator. Furthermore, since superposition of the single quantum with the vacuum state is not required \cite{Dunningham2007,Fuwa2015,Lee2017,Greenberger1995,Hardy1994,Hardy1995,Tan1991,Vaidman1995}, the experiment is not subject to superselection rules and applies to massive particles as well as photons.

If the quantum particle moves at speed $v$ and the interferometer arms are of length $L$, to ensure lack of signaling between Alice and Bob at a distance $\sqrt{2}L$, all detectors
$D_{i}$ and Alice's device $M$ have to be strongly measured
at a time $t = L/v$ after Bob's action, which is
$\Delta t = (\frac{\sqrt{2}}{c}-\frac{1}{v})L$ ahead of time before any signal from Bob
could reach Alice at luminal speed. To ensure $\Delta t > 0$, the quantum particle has to move at speed $v>\frac{1}{\sqrt{2}}c$.

Since we are interested in different post-selections, we will use superscripts
to denote the quantum probability amplitudes $\psi_{j}^{(i)}$ propagating
along different quantum histories $\hat{D}_{i}\odot\hat{x}_{j}\odot\hat{S}$,
where $i\in\{1,2,3\}$ denotes the detectors~$D_{i}$,
$j\in\{1,2,3\}$ denotes the interferometer arms~$x_{j}$, and $\hat{S}$ denotes the source of single particles. Thus,
there is a one-to-one correspondence
\begin{equation}
\psi_{j}^{(i)}\qquad\leftrightarrow\qquad\hat{D}_{i}\odot\hat{x}_{j}\odot\hat{S}.
\end{equation}
Also, we will denote the weak value of the projector $\hat{x}_{j}$
for post-selected $D_{i}$ as $(x_{j})_{w}^{(i)}$, which is given by
the ratio of the quantum probability amplitude for the individual
Feynman history through $x_{j}$ and the total Feynman sum from $S$
to $D_{i}$ (cf. Theorem 8 in Ref.~\refcite{Georgiev2018})
\begin{equation}
(x_{j})_{w}^{(i)}=\frac{\psi_{j}^{(i)}}{\sum_{j'}\psi_{j'}^{(i)}}.
\label{eq:weak}
\end{equation}

\begin{figure}[t]
\begin{centering}
\includegraphics[width=80mm]{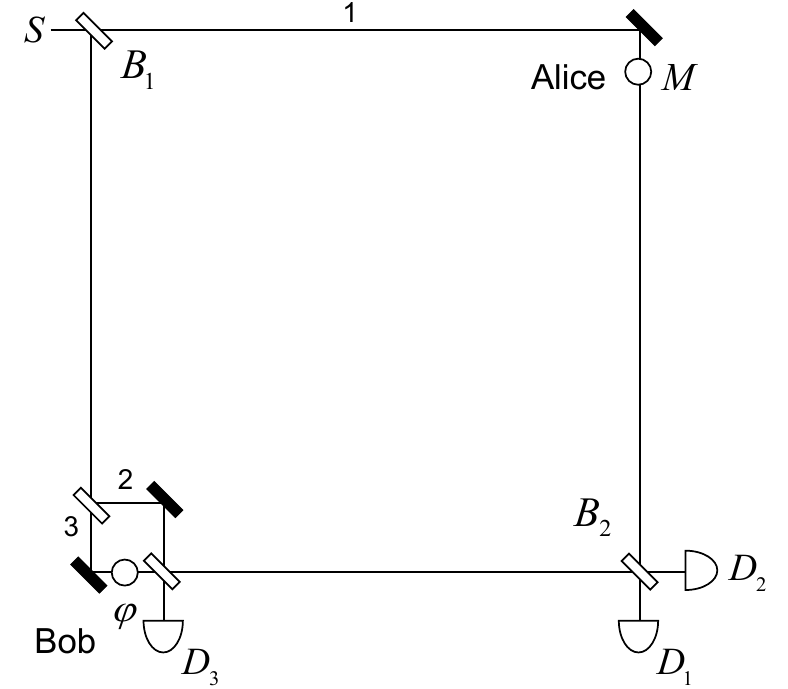}
\par\end{centering}
\caption{\label{fig:1}A square nested Mach--Zehnder interferometer in which a single quantum particle can travel from
the source $S$ to three detectors $D_{i}$.
On path~1, Alice measures an observable of the meter pointer that is a function of the weak value $(x_{1})_{w}$, while on path~3, Bob controls a phase shifter~$\varphi$.}
\end{figure}

\subsection{Post-selection at \texorpdfstring{$D_{1}$}{D1}}

For the calculation of the weak value $(x_{1})_{w}^{(1)}$, there are only 3 continuous quantum histories
from the source $S$ to detector $D_{1}$ that
need to be summed over
\begin{equation}
\psi_{1}^{(1)}=\frac{1}{2}\imath;\qquad\psi_{2}^{(1)}=\frac{1}{4}\imath;\qquad\psi_{3}^{(1)}=-\frac{1}{4}\imath e^{\imath\varphi}.
\end{equation}
The weak value
to be measured by Alice depends nonlocally on Bob's choice
\begin{equation}
(x_{1})_{w}^{(1)}=\frac{\psi_{1}^{(1)}}{\psi_{1}^{(1)}+\psi_{2}^{(1)}+\psi_{3}^{(1)}}=\frac{\frac{1}{2}\imath}{\frac{1}{2}\imath+\frac{1}{4}\imath-\frac{1}{4}\imath e^{\imath\varphi}}=\frac{2}{3-e^{\imath\varphi}};\label{eq:Aw-D1}
\end{equation}
\begin{gather}
\textrm{Re}\left[(x_{1})_{w}^{(1)}\right]=\frac{3-\cos\varphi}{5-3\cos\varphi};\quad
\textrm{Im}\left[(x_{1})_{w}^{(1)}\right]=\frac{\sin\varphi}{5-3\cos\varphi};\quad
|(x_{1})_{w}^{(1)}|^{2}=\frac{2}{5-3\cos\varphi}.
\end{gather}
In the absence of the weak measuring device, detector $D_{1}$ clicks
with probability $P_{1}=\frac{1}{8}(5-3\cos\varphi)$.
In the presence of the weak measuring device this probability is modified
by a certain amount given in Eq.~\eqref{eq:Prob-D1}.

\subsection{Post-selection at \texorpdfstring{$D_{2}$}{D2}}

There are only 3 continuous quantum histories from the source $S$ to
detector $D_{2}$ that need to be summed over
\begin{equation}
\psi_{1}^{(2)}=-\frac{1}{2};\qquad\psi_{2}^{(2)}=\frac{1}{4};\qquad\psi_{3}^{(2)}=-\frac{1}{4} e^{\imath\varphi}.
\end{equation}
The weak value again depends
nonlocally on Bob's choice, but solely due to its imaginary part
\begin{equation}
(x_{1})_{w}^{(2)}=\frac{\psi_{1}^{(2)}}{\psi_{1}^{(2)}+\psi_{2}^{(2)}+\psi_{3}^{(2)}}=\frac{-\frac{1}{2}}{-\frac{1}{2}+\frac{1}{4}-\frac{1}{4}e^{\imath\varphi}}=\frac{2}{1+e^{\imath\varphi}};\label{eq:Aw-D2}
\end{equation}
\begin{gather}
\textrm{Re}(x_{1})_{w}^{(2)}=1;\quad
\textrm{Im}(x_{1})_{w}^{(2)}=-\frac{\sin\varphi}{1+\cos\varphi};\quad
|(x_{1})_{w}^{(2)}|^{2}=\frac{2}{1+\cos\varphi}.
\end{gather}
In the absence of the weak measuring device, detector $D_{2}$ clicks
with probability $P_{2}=\frac{1}{8}(1+\cos\varphi)$.
In the presence of the weak measuring device this probability is modified
by a certain amount given in Eq.~\eqref{eq:Prob-D2}.

\subsection{Post-selection at \texorpdfstring{$D_{3}$}{D3}}

Because the quantum history through $x_{1}$ contributes zero quantum probability amplitude to $D_3$,
the weak value is always zero
\begin{equation}
(x_{1})_{w}^{(3)} = 0 . \label{eq:Aw-D3}
\end{equation}
Detector $D_{3}$ clicks with probability
$P_3=\frac{1}{4}(1+\cos\varphi)$ with or without the weak measuring device.

Let the experiment be repeated many times, using a single quantum particle
in the interferometer per each run. We are interested to find out
whether it is possible to explain the observed outcomes using
any local hidden variables model of quantum mechanics.

\section{Quantum measurements of pointer observables dependent on weak values}

Before we proceed further, it is important to know the quantum distributions for
pointer observables of Alice's weak measuring device $M$, which are functionally dependent on the weak
values of quantum histories for different post-selections of the single quantum particle traversing the interferometer.

Alice's measuring device $M$ starts with a real-valued Gaussian position wavefunction
centered at zero
\begin{equation}
\phi_{0}(x)=\frac{1}{(2\pi\sigma^2)^{\frac{1}{4}}} e^{-\frac{x^{2}}{4\sigma^{2}}}. \label{eq:meter-0}
\end{equation}

The interaction Hamiltonian between the
single quantum particle $S$ and the measuring device $M$ is
\begin{equation}
\hat{H}_{\textrm{int}}=g\delta(t-t_{m})\,\hat{A}\otimes\hat{p}_x,
\end{equation}
where $\hat{A}$ is an observable for the measured quantum particle~$S$ and
$\hat{p}_x=\hbar\hat{k}_x$
is the meter variable conjugate to the meter pointer variable~$\hat{x}$.
Further, the quantum particle~$S$ evolves with internal Hamiltonian
$\hat{H}_{S}\otimes\hat{I}_M$, while the internal Hamiltonian
of the meter is suppressed $\hat{I}_S\otimes\hat{H}_{M}=0$.

The composite system starts from the initial state
\begin{equation}
|\psi_{i}\rangle|\phi_{0}\rangle=|\psi_{i}\rangle \int_{-\infty}^{\infty}\phi_{0}(x)|x\rangle \, dx
\end{equation}
and evolves in time with the operator
\begin{equation}
\hat{\mathcal{T}}_{f,m}\,e^{-\frac{\imath}{\hbar}g\,\hat{A}\otimes\hat{p}_{x}}\hat{\mathcal{T}}_{m,i}
\end{equation}
where $\hat{\mathcal{T}}_{b,a}=e^{-\frac{\imath}{\hbar}\int_{t_{a}}^{t_{b}}\hat{H}_{S}\otimes\hat{I}_{M}\,dt}$
are internal time evolution operators of the measured quantum particle~$S$.

For the post-selected system in a final state $|\psi_{f}\rangle$,
the projected (not normalized) final meter wavefunction in the position
basis is
\begin{equation}
\phi_{f}(x)=
\frac{\langle\psi_{f}|\hat{\mathcal{T}}_{f,i}|\psi_{i}\rangle}{(2\pi\sigma^{2})^{\frac{1}{4}}}\left[\left(1-A_{w}\right)e^{-\frac{x^{2}}{4\sigma^{2}}}+A_{w}e^{-\frac{(x-g)^{2}}{4\sigma^{2}}}\right]
\label{eq:wavefunc-x}
\end{equation}
where $A_{w}$ is the weak value defined in Eq.~\eqref{eq:weak}, and we have utilized the fact that the observable of interest is a projection operator represented by an idempotent matrix, $\hat{A}^{2}=\hat{A}$. (For a complete analytic derivation of the latter formula see \ref{app:analytic}.)

After applying the Born rule, the corresponding final meter probability density distribution is
\begin{align}
\Phi_{f}(x)=\left|\phi_{f}(x)\right|^{2}=& \left|\langle\psi_{f}|\hat{\mathcal{T}}_{f,i}|\psi_{i}\rangle\right|^{2}\frac{1}{\sqrt{2\pi\sigma^{2}}}e^{-\frac{x^{2}}{2\sigma^{2}}}\nonumber\\
&\times\left[1-2\textrm{Re}\left(A_{w}\right)\left(1-e^{\frac{x^{2}-(x-g)^{2}}{4\sigma^{2}}}\right)+\left|A_{w}\right|^{2}\left(1-e^{\frac{x^{2}-(x-g)^{2}}{4\sigma^{2}}}\right)^{2}\right].
\label{eq:Phi-f-x}
\end{align}

If Alice measures $M$ in the wavenumber basis, the final meter wavefunction is Fourier transformed into
\begin{equation}
\phi_{f}(k)=
\langle\psi_{f}|\hat{\mathcal{T}}_{f,i}|\psi_{i}\rangle\left(\frac{2}{\pi}\sigma^{2}\right)^{\frac{1}{4}}e^{-k^{2}\sigma^{2}}\left[1-A_{w}+A_{w}e^{-\imath gk}\right]
\label{eq:wavefunc-k}
\end{equation}
with corresponding final meter probability density distribution
\begin{align}
\Phi_{f}(k)=\left|\phi_{f}(k)\right|^{2}=&\left|\langle\psi_{f}|\hat{\mathcal{T}}_{f,i}|\psi_{i}\rangle\right|^{2}\sqrt{\frac{2}{\pi}}\sigma e^{-2k^{2}\sigma^{2}}\nonumber\\
&\times\left\{ 1+2\textrm{Im}\left(A_{w}\right)\sin\left(gk\right)-2\left[\textrm{Re}\left(A_{w}\right)-\left|A_{w}\right|^{2}\right]\left[1-\cos\left(gk\right)\right]\right\} .
\label{eq:Phi-f-k}
\end{align}

The probability of post-selection of the quantum particle $S$ in a final meter state $|\psi_{f}\rangle$
at detector $D_{f}$ is affected by a certain amount due to the weak
interaction with the measuring device $M$ as follows
\begin{align}
\textrm{Prob}\left(D_{f}\right)=&\int_{-\infty}^{\infty}\left|\phi_{f}(x)\right|^{2}\,dx=\int_{-\infty}^{\infty}\left|\phi_{f}(k)\right|^{2}\,dk\nonumber\\
=&~\left|\langle\psi_{f}|\hat{\mathcal{T}}_{f,i}|\psi_{i}\rangle\right|^{2}\left\{ 1-2\left[\textrm{Re}\left(A_{w}\right)-\left|A_{w}\right|^{2}\right]\left(1-e^{-\frac{g^{2}}{8\sigma^{2}}}\right)\right\}
\label{eq:prob}
\end{align}
where $|\langle\psi_{f}|\hat{\mathcal{T}}_{f,i}|\psi_{i}\rangle|^{2}$
is the probability for the quantum particle to end at detector $D_{f}$
in the interferometer in the absence of weak measuring device and
the second factor is due to quantum interference of the pointer position
involving the corresponding weak value.

\section{Local hidden variables model}
\label{sec:LHV}
Einstein's principle of locality states that an action performed on a system $S_1$ must not modify the physical description of another system $S_2$ for any two physical systems that are spacelike separated. Thus, an effect cannot occur from a cause that is not in its past
light cone. Similarly, a cause cannot have an effect
outside its future light cone. A physical model is \emph{local} if
it satisfies Einstein's principle of locality.

Relativistic classical particles
can travel along a single path, but not along two or more paths at the same time.
In order to reproduce the correct quantum distributions that depend on the weak values of corresponding quantum histories, the particle needs to visit both Alice (to affect her weak measuring device $M$) and Bob (to obtain information about his choice of phase shifter $\varphi$). Since in our setup the particle does not have the time to visit at a luminal speed both Alice and Bob along a single path, the particle would necessarily fail to reproduce correctly all the observable quantum outcomes without invoking some additional hidden signaling.

With hidden signaling, the particle may receive new information while traveling inside the interferometer and
may use a list of contextual instructions for navigation.
To make local choices, however, the particle has to rely on classical mixtures of statistical distributions.

To proceed, we consider a general local hidden variables model characterized by the following three properties (for a detailed discussion of each property see \ref{app:Einstein} and \ref{app:Wheeler}):

(1) The particle and Alice's measuring device possess a probabilistic
mechanism that could generate an outcome drawn from any given statistical distribution
$\Lambda$.

(2) If hidden signals exist, they travel at most at luminal speed
and cannot have any physical effects outside their future light cones.

(3) The particle possesses memory and executes a list of contextual instructions,
which allow usage of new information obtained through hidden signals.

In addition, we arrange the setup so that Bob is able to choose the setting of the phase shifter~$\varphi$ in a delayed fashion only after the particle has passed the first beam splitter $B_1$.

In a spacelike separated manner from Bob's action, Alice performs her projective measurement of the weak measuring device $M$ shortly after the single quantum has passed the second beam splitter $B_2$, but before hidden signaling from $B_2$ could reach $M$. If Alice is given the choice to occasionally block completely her interferometer arm, she will detect the quantum particle with probability of $\frac{1}{2}$. Therefore, in order to be consistent with the quantum mechanical predictions any
local hidden variables model should predict that at the first beam splitter~$B_{1}$,
the particle goes with equal probability toward either Alice or Bob.

\section{Constraints imposed by Alice's choice\label{sec:Alice}}

Because Alice chooses which observable ($\hat{x}$ or $\hat{k}$) to measure on $M$ only after the single quantum has passed through all beam splitters, the local hidden variables model is constrained to operate without knowledge of Alice's choice. In such case, the probabilistic mechanism possessed by quantum systems in the local hidden variables model could operate in two modes of commitment:

(1) Commitment to a distribution $\Lambda$ in case some particular observable is measured.

(2) Commitment to an outcome $\lambda$ drawn from distribution $\Lambda$ in case some particular observable is measured.

We will analyze each of these two possibilities in turn.

\subsection{Case \texorpdfstring{$\varphi=0$}{zero phi}}

The initial setting of the interferometer introduces no phase shift,
$\varphi=0$, at Bob's location. If Bob does not change the phase shifter $\varphi$, the quantum could reach
detector $D_{3}$ with probability of $\frac{1}{2}$ while the weak
measuring device is committed to one of two incompatible initial probability distributions (default Bob's distributions)
\begin{align}
\Phi_{B}(x) =& \frac{1}{\sqrt{2\pi\sigma^{2}}}e^{-\frac{x^{2}}{2\sigma^{2}}}, \\
\Phi_{B}(k) =& \sqrt{\frac{2}{\pi}}\sigma e^{-2k^{2}\sigma^{2}},
\end{align}
 which is going to be generated at the time when Alice chooses to actually measure $\hat{x}$ or $\hat{k}$.

Alternatively, the quantum could reach detectors $D_{1}$ or $D_{2}$
with probability of $\frac{1}{4}$ each, while the weak measuring
device is committed to change conditionally on Alice's choice of a measurement basis into one of the following final probability distributions (default Alice's distributions)
\begin{align}
\Phi_{A}(x) =& \frac{1}{\sqrt{2\pi\sigma^{2}}}e^{-\frac{\left(x-g\right)^{2}}{2\sigma^{2}}},\\
\Phi_{A}(k) =& \sqrt{\frac{2}{\pi}}\sigma e^{-2k^{2}\sigma^{2}}.
\end{align}

The local hidden variables model is able to reproduce the correct quantum
distributions through hidden signaling obtained from the point
of bifurcation at the first beam splitter $B_{1}$ as follows: To replicate
correctly the quantum probabilities, at $B_{1}$ the particle has
to go with equal probability to one of the two interferometer arms
that lead to Alice or Bob. If at $B_{1}$ the particle goes to Bob,
it has to exit always at $D_{3}$ to account for $\textrm{Prob}\left(D_{3}\right)=\frac{1}{2}$.
Alice's device $M$ is cued by the particle absence to make a commitment to select probabilistically
an outcome from the initial distributions $\Phi_{B}$ upon measurement of $\hat{x}$ or $\hat{k}$. If at $B_{1}$
the particle goes to Alice, the device $M$ is cued by the particle
presence to make a commitment to select probabilistically an outcome from the shifted distributions
$\Phi_{A}$ upon measurement of $\hat{x}$ or $\hat{k}$, after which the particle goes to $B_{2}$, where it is
reflected with equal probability to either $D_{1}$ or $D_{2}$ to account
for $\textrm{Prob}\left(D_{1}\right)=\textrm{Prob}\left(D_{2}\right)=\frac{1}{4}$.
Thus, the local hidden variables model reproduces exactly the correct quantum
distributions if the quantum experiment is performed without
any action by Bob.

\subsection{Case \texorpdfstring{$\varphi \neq 0$}{non-zero phi}}

The experimental setup is arranged so that Alice and Bob could choose to completely block
their interferometer arms in a delayed choice fashion after the quantum particle
has passed $B_{1}$. From coincident measurements when both Alice
and Bob block their corresponding arms, it could be established that
at $B_{1}$ the quantum particle goes with probability $\frac{1}{2}$ to Alice
and $\frac{1}{2}$ to Bob. Because for $\varphi=0$ the distributions
of the weak measuring device $\Phi_{A}$ and $\Phi_{B}$ are perfectly
correlated with the beam splitting of ``non-empty waves'' for the single quantum particle at $B_{1}$,
the local hidden variables model should use hidden signaling (``empty
waves'') propagated from $B_{1}$ to $M$ in order to select from $\Phi_{A}$
or $\Phi_{B}$ in perfect correlation with the quantum particle (``non-empty
wave'') traversing Alice's or Bob's interferometer arm, respectively.
To enforce locality, ``empty waves'' can never produce a quantum
particle, whereas the ``non-empty wave'' always produce the quantum
particle upon measurement (for details, see \ref{app:nonlocal}, \ref{app:Einstein} and \ref{app:Wheeler}).

If Bob introduces a non-zero phase shift, $\varphi\neq0$, in a delayed choice fashion such that
the particle has already passed $B_{1}$ and Bob's choice is spacelike
separated from the final projective measurement of the weak measuring
device $M$, at the second beam splitter $B_{2}$ the local hidden variables model could deliver remote information about~$\varphi$ (obtained through hidden signaling from Bob to~$B_2$) to operate on a statistical mixture of $\Phi_{A}$~and~$\Phi_{B}$ due to contribution of ``non-empty waves'' from Bob's arm:
The quanta that arrive at detector $D_{3}$ with probability of $\frac{1}{4}(1+\cos\varphi)$
are perfectly correlated with $\Phi_{B}$ thereby reproducing correctly
the quantum outcomes for post-selected $D_{3}$. Because $\frac{1}{2}$
of ``non-empty waves'' traverse Bob's interferometer arm, at $B_{2}$
arrive quanta correlated with $\Phi_{A}$ with probability of $\frac{1}{2}$ and quanta correlated with $\Phi_{B}$ with probability $\frac{1}{2}-\frac{1}{4}(1+\cos\varphi)=\frac{1}{4}(1-\cos\varphi)$.
The hidden signaling received at $B_{2}$ provides remote information
for both Bob's choice $\varphi$ and the prior commitment to $\Phi_i$ at $M$ allowing for
statistical mixing of $\Phi_{A}$ and $\Phi_{B}$.

In essence, at the second beam splitter $B_2$ the local hidden variables model will have to use the information
for Bob's choice $\varphi$ in order to split single quanta with the following distributions
\begin{eqnarray}
\Phi_{+}(x) & = & \frac{1}{2}\Phi_{A}(x)+\frac{1}{4}\left(1-\cos\varphi\right)\Phi_{B}(x), \label{eq:AB-1}\\
\Phi_{+}(k) & = & \frac{1}{2}\Phi_{A}(k)+\frac{1}{4}\left(1-\cos\varphi\right)\Phi_{B}(k). \label{eq:AB-2}
\end{eqnarray}
This is consistent with the correct quantum distributions as a consequence of the quantum no-communication theorem \cite{Eberhard1978,Ghirardi1980,dEspagnat2003,Peres2004}. Indeed, from Eqs.~\eqref{eq:Phi-f-x} and~\eqref{eq:Phi-f-k} (see also \ref{app:distributions}) it can be directly verified that
\begin{eqnarray}
\Phi_{+}(x) & = & \Phi_{1}(x)+\Phi_{2}(x), \label{Phi+x}\\
\Phi_{+}(k) & = & \Phi_{1}(k)+\Phi_{2}(k), \label{Phi+k}
\end{eqnarray}
Because Alice's choice of measurement basis ($x$ or $k$) is not available at $B_2$, however, the local hidden variables model may attempt to split the available quantum distributions before $B_2$, $\Phi_{+}(x)$ or $\Phi_{+}(k)$, into the correct quantum distributions after $B_2$, $\Phi_{1}(x)$ or $\Phi_{1}(k)$ for~$D_1$ and $\Phi_{2}(x)$ or $\Phi_{2}(k)$ for~$D_2$, using one of the following two modes of commitment.

\subsubsection{Commitment to a distribution}

Because Alice's measuring device $M$ does not have access to Bob's choice $\varphi$, it has to make a commitment for future action, which is then sent to the second beam splitter $B_2$. If $M$ is committed only to a distribution $\Phi_A$ or $\Phi_B$, the task at $B_2$ would be to prepare the correct quantum distributions for $D_1$ or $D_2$ as a convex combination of $\Phi_A$ and $\Phi_B$. This task, however, cannot be achieved for $\varphi\neq 0$.

\begin{theorem}
Let a normalized distribution $\Phi_i (\lambda)$ be a convex combination of two other normalized distributions $\Phi_A (\lambda)$ and $\Phi_B (\lambda)$ expressed as
\begin{equation}
\Phi_i (\lambda) = w_A \Phi_A (\lambda) + w_B \Phi_B (\lambda) \label{eq:convex}
\end{equation}
where $w_A\geq 0$, $w_B\geq 0$, and $w_A + w_B = 1$. Then, the corresponding weights are constants determined by
\begin{eqnarray}
w_A &=& \frac{\Phi_i (\lambda) - \Phi_B (\lambda)}{\Phi_A (\lambda) - \Phi_B (\lambda)}, \label{eq:wA}\\
w_B &=& \frac{\Phi_i (\lambda) - \Phi_A (\lambda)}{\Phi_B (\lambda) - \Phi_A (\lambda)}. \label{eq:wB}
\end{eqnarray}
\end{theorem}
\begin{proof}
Substitution of $w_B = 1 -w_A$ in Eq.~\eqref{eq:convex} followed by algebraic rearrangement gives Eq.~\eqref{eq:wA}. Similarly, Eq.~\eqref{eq:wB} is obtained using $w_A = 1 -w_B$ in Eq.~\eqref{eq:convex}.
\end{proof}

\begin{figure}[t]
\begin{centering}
\includegraphics[width=125mm]{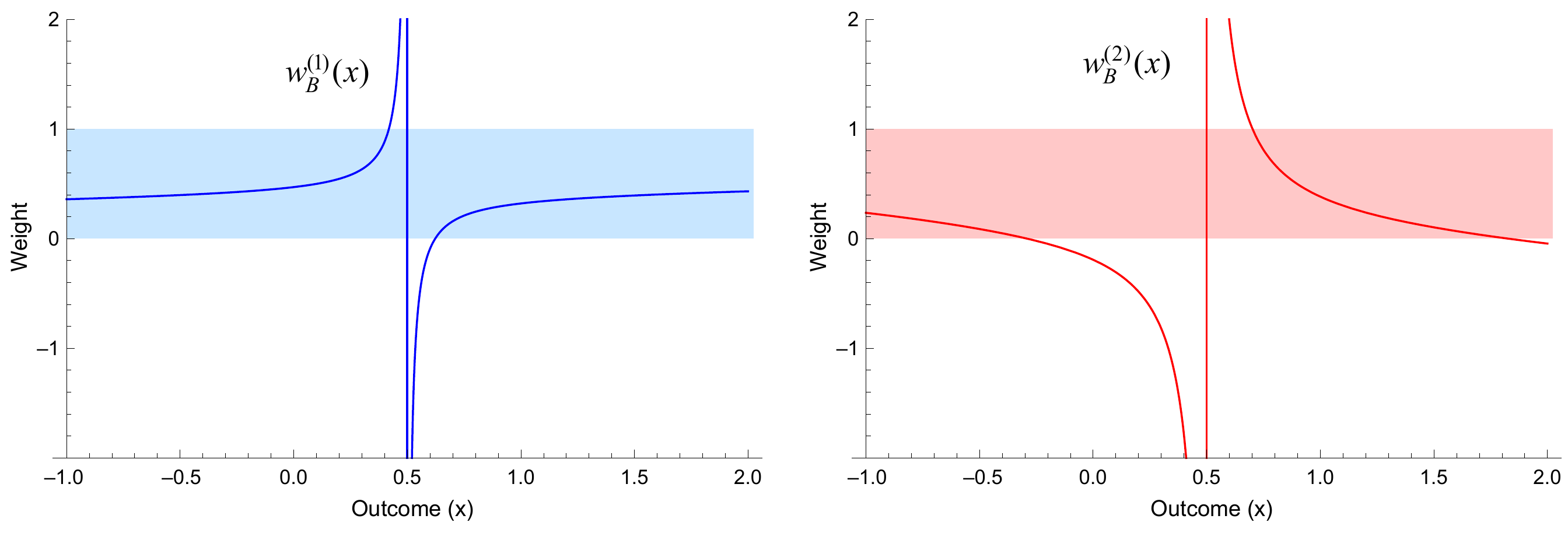}
\par\end{centering}
\caption{\label{fig:2} Plots of $w_B^{(1)} (x)$ and $w_B^{(2)} (x)$ for $g=1$, $\sigma=1$ and $\varphi=\frac{\pi}{2}$. The regions between $0$ and $1$ indicate weights accessible by local hidden variables models through convex combination of default distributions. Critical tests of such local hidden variables models could be performed for $x$ outcomes whose weights are outside the region $[0,1]$.}
\end{figure}

The normalized quantum distributions $\tilde{\Phi}_1 (x)$, $\tilde{\Phi}_2 (x)$, $\tilde{\Phi}_1 (k)$ or $\tilde{\Phi}_2 (k)$
for post-selected detectors $D_1$ or $D_2$ that need to be reproduced by the local hidden variables model are given in \ref{app:distributions}. To experimentally rule out the hidden variables model, it is sufficient to identify regions with quantum measurement outcomes for which the weights are outside the admissible region $\left[0,1\right]$.

For measurement in $x$-basis, $w_A^{(1)} (x) = \frac{\tilde{\Phi}_1 (x) - \Phi_B (x)}{\Phi_A (x) - \Phi_B (x)}$ contains a region of $x$ outcomes with $w_A^{(1)} (x) >1$ for $\varphi\in (0,2\pi)$ and a region of $x$ outcomes with $w_A^{(1)} (x)<0$ for $\varphi\in (0.11,2\pi-0.11)$. Similarly, $w_A^{(2)} (x) = \frac{\tilde{\Phi}_2 (x) - \Phi_B (x)}{\Phi_A (x) - \Phi_B (x)}$ contains a region of $x$ outcomes with $w_A^{(2)} (x)>1$ for $\varphi\in (0,2\pi)$ and a region of $x$ outcomes with $w_A^{(2)} (x)<0$ for $\varphi\in (0.065,2\pi-0.065)$. Converse results with respect to $0$ and $1$ hold for $w_B^{(1)} (x) = 1-w_A^{(1)}(x)$ and $w_B^{(2)} (x) = 1-w_A^{(2)}(x)$ (see Fig.~\ref{fig:2} for a numerical example).

For measurement in $k$-basis, all weights are undefined due to division by zero resulting from $\Phi_{A}(k)=\Phi_{B}(k)$. In other words, it is impossible to obtain shifted distributions $\Phi_{1}(k)$ or $\Phi_{2}(k)$ through convex mixing of the initial distribution with itself.

\subsubsection{Commitment to an outcome}

The remaining alternative that could be attempted by the local hidden variables model is to commit $M$ for each run $j$ to a particular outcome $x_j$ or $k_j$ drawn from the correct quantum distributions $\Phi_{+}(x)$ or $\Phi_{+}(k)$.

If Alice is limited to a single measurement of the same observable, the local hidden variables model could easily use the weights of splitting towards the two detectors in order to prepare the correct distributions $\Phi_1$ and~$\Phi_2$.

As an example, suppose that Alice always measures the observable $\hat{x}$. The weights of splitting towards $D_1$ or $D_2$ at the second beam splitter $B_2$ are
\begin{equation}
w_{1}(x)=\frac{\Phi_{1}(x)}{\Phi_{+}(x)},\quad\quad w_{2}(x)=\frac{\Phi_{2}(x)}{\Phi_{+}(x)}. \label{w-x}
\end{equation}
Because the weights of splitting are not constant, but exhibit functional dependence on the outcome $x$, the prior commitment of $M$ to generate particular $x_j$ would provide the required information at $B_2$ for correct reproduction of the quantum distributions $\Phi_1(x)$ and $\Phi_2(x)$.

Similarly, if Alice always measures the observable~$\hat{k}$, the prior commitment of $M$ to generate $k_j$ would provide the required information at $B_2$ for splitting with the correct weights
\begin{equation}
w_{1}(k)=\frac{\Phi_{1}(k)}{\Phi_{+}(k)},\quad\quad w_{2}(k)=\frac{\Phi_{2}(k)}{\Phi_{+}(k)}. \label{w-k}
\end{equation}

To study quantum nonlocality in our setup, however, we have allowed Alice to choose which observable to measure, $\hat{x}$ or $\hat{k}$, only after the single quantum particle has passed the second beamplitter $B_2$. The local hidden variables model is thereby forced to attempt matching $x_j$ and $k_j$ outcomes based on their weights of splitting.

\begin{figure*}[t]
\begin{centering}
\includegraphics[width=125mm]{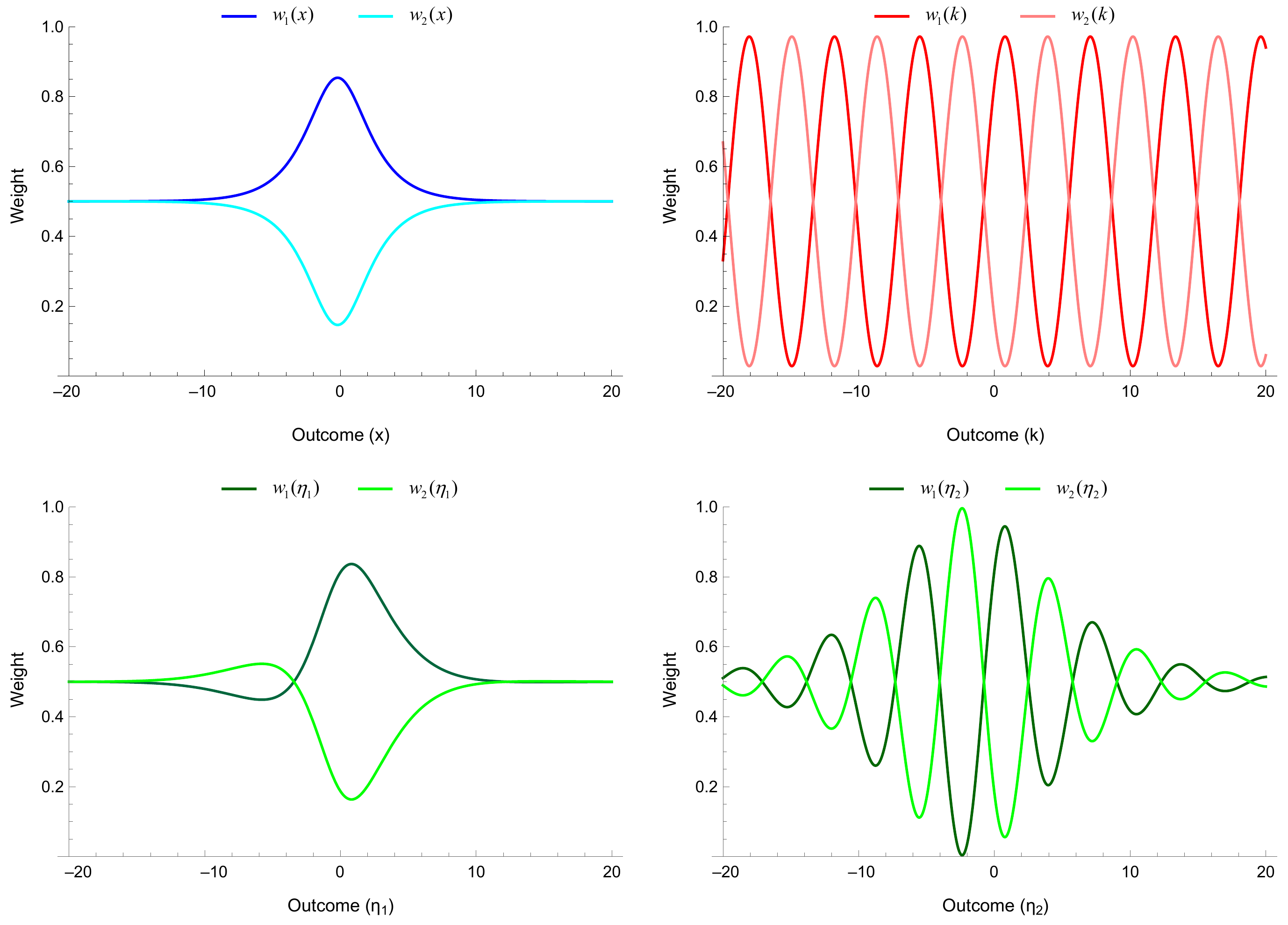}
\par\end{centering}
\caption{Plots of the quantum weights for splitting at $B_2$ depending on Alice's choice of measurement of $\hat{x}$, $\hat{k}$, $\hat{\eta_1}=1\cdot\hat{x}+1\cdot\hat{k}$ or $\hat{\eta_2}=0.1\cdot\hat{x}+1\cdot\hat{k}$ for $g=1$, $\sigma=1$ and $\varphi=\frac{\pi}{2}$.}
\label{fig:3}
\end{figure*}

The most general local hidden variables model should produce both $x$ and $k$ quantum distributions using at most two different hidden functions $f_A(x,k)$ and $f_B(x,k)$ which are independent of Bob's choice $\varphi$, but may take into account whether the particle goes to Alice or Bob at the first beam splitter $B_1$. Each hidden function  $f_A(x,k)\geq0$ or $f_B(x,k)\geq0$ gives the probability density for commitment to the particular outcomes $x$ and $k$ (one of which will be revealed depending on Alice's future choice) provided that the quantum particle is reflected towards Alice or Bob at $B_1$, respectively.
Both hidden functions $f_{A}(x,k)$ and $f_{B}(x,k)$ should be normalized bivariate probability density distributions such that
\begin{eqnarray}
\int_{-\infty}^{\infty}f_{A}(x,k)dx & = & \Phi_{A}(k), \label{model-1}\\
\int_{-\infty}^{\infty}f_{A}(x,k)dk & = & \Phi_{A}(x), \label{model-2}\\
\int_{-\infty}^{\infty}f_{B}(x,k)dx & = & \Phi_{B}(k), \label{model-3}\\
\int_{-\infty}^{\infty}f_{B}(x,k)dk & = & \Phi_{B}(x). \label{model-4}
\end{eqnarray}
From Eqs.~\ref{eq:AB-1} and \ref{eq:AB-2}, it follows that $f_A(x,k)$ and $f_B(x,k)$ should contribute single quantum particles at the second beam splitter $B_2$ respectively with probabilities $\frac{1}{2}$ and $\frac{1}{4}(1-\cos{\varphi})$ because these are the expected quantum probabilities of detecting the single quantum particle if the detectors are placed on Alice's or Bob's arms immediately before~$B_2$.
The splitting at $B_2$ towards $D_1$ or $D_2$ could then occur with at most two different sets of weight functions
\begin{gather}
0\leq w_1^{A}(x,k,\varphi)\leq 1, \label{model-5}\\
0\leq w_1^{B}(x,k,\varphi)\leq 1, \label{model-6}\\
w_2^{A}(x,k,\varphi) = 1 - w_1^{A}(x,k,\varphi), \label{model-7}\\
w_2^{B}(x,k,\varphi) = 1 - w_1^{B}(x,k,\varphi), \label{model-8}
\end{gather}
which could use all the available information including Bob's choice $\varphi$ in addition to the prior commitment to outcomes $x$ and $k$.
To reproduce correctly the quantum distributions $\Phi_{i}\left(x,\varphi\right)$ and $\Phi_{i}\left(k,\varphi\right)$ at $B_2$, the local hidden variables model should satisfy the following integral equations
\begin{align}
\Phi_{i}\left(x,\varphi\right)
 = & ~\int_{-\infty}^{\infty}\left[\frac{1}{2} f_{A}(x,k)\,w_{i}^{A}\left(x,k,\varphi\right)
+\frac{1}{4}\left(1-\cos\varphi\right)f_{B}(x,k)\,w_{i}^{B}\left(x,k,\varphi\right)\right]dk,
\label{eq:LHV-1}\\
\Phi_{i}\left(k,\varphi\right)
 = & ~\int_{-\infty}^{\infty}\left[\frac{1}{2} f_{A}(x,k)\,w_{i}^{A}\left(x,k,\varphi\right)
+\frac{1}{4}\left(1-\cos\varphi\right)f_{B}(x,k)\,w_{i}^{B}\left(x,k,\varphi\right)\right]dx.
\label{eq:LHV-2}
\end{align}
The RHS of Eqs.~\eqref{eq:LHV-1} and \eqref{eq:LHV-2} can be identified with integrals of the Wigner function as follows
\begin{align}
\Phi_{i}\left(x,\varphi\right)= & \int_{-\infty}^{\infty}W_{i}(x,k,\varphi)\,dk,\label{eq:LHV-3}\\
\Phi_{i}\left(k,\varphi\right)= & \int_{-\infty}^{\infty}W_{i}(x,k,\varphi)\,dx.\label{eq:LHV-4}
\end{align}
where the Wigner function is given by
\begin{equation}
W_{i}(x,k) = \frac{1}{\pi}\int_{-\infty}^{\infty}\phi_{i}^{*}(x+z)\phi_{i}(x-z)e^{2\imath kz}dz
= \frac{1}{\pi}\int_{-\infty}^{\infty}\phi_{i}^{*}(k+z)\phi_{i}(k-z)e^{-2\imath xz}dz \label{eq:Wigner}
\end{equation}
A possible way to identify the integrand in the local hidden variables model (Eqs.~\ref{eq:LHV-1} and \ref{eq:LHV-2}) with the Wigner function (Eqs.~\ref{eq:LHV-3} and \ref{eq:LHV-4}) is to allow Alice the choice to measure any linear combination of position and momentum (also called rotated quadrature).

If Alice measures the quantum operator $\hat{\eta}=a\hat{x}+b\hat{k}$, its
eigenstates and eigenvalues are given by
\begin{equation}
\hat{\eta}|\eta\rangle=\eta|\eta\rangle
\end{equation}
The tomogram of the quantum state $|\phi\rangle$ for the rotated quadrature $\eta$ is then given by
the Radon transformation of the Wigner function $W(x,k)$ \cite{Wang2017,deGosson2017}
\begin{equation}
\left|\langle\phi|\eta\rangle\right|^{2}=\int_{-\infty}^{\infty}\int_{-\infty}^{\infty}\delta\left(\eta-ax-bk\right)W(x,k)\,dxdk
\end{equation}
The observed probability distribution of $\eta$ outcomes can be directly related to the position and momentum wavefunctions of the meter \cite{Wang2017}
\begin{equation}
\left|\langle\phi|\eta\rangle\right|^{2}= \frac{1}{2\pi}\sqrt{\frac{1}{\imath ab}}\int_{-\infty}^{\infty}\int_{-\infty}^{\infty}\phi^{*}(k)\phi(x)
e^{-\frac{\imath}{2ab}\left[\eta^{2}-\left(\eta-ax\right)^{2}-\left(\eta-bk\right)^{2}\right]}\,dxdk.
\end{equation}
Then, given that Alice's measurements agree with the quantum distributions, one can use the Wigner Uniqueness theorem by Blass and Gurevich \cite{Blass2015}, which restates the tomographic characterization of Wigner's function due to Jacqueline and Pierre Bertrand~\cite{Bertrand1987}.

\begin{theorem}
(Wigner Uniqueness). Wigner's quasi-distribution defined by Eq.~\eqref{eq:Wigner} is the unique function on the phase space that yields the correct marginal distributions not only for position and momentum but for all their linear combinations. \cite{Blass2015,Bertrand1987}
\end{theorem}

\begin{figure*}[t]
\begin{centering}
\includegraphics[width=125mm]{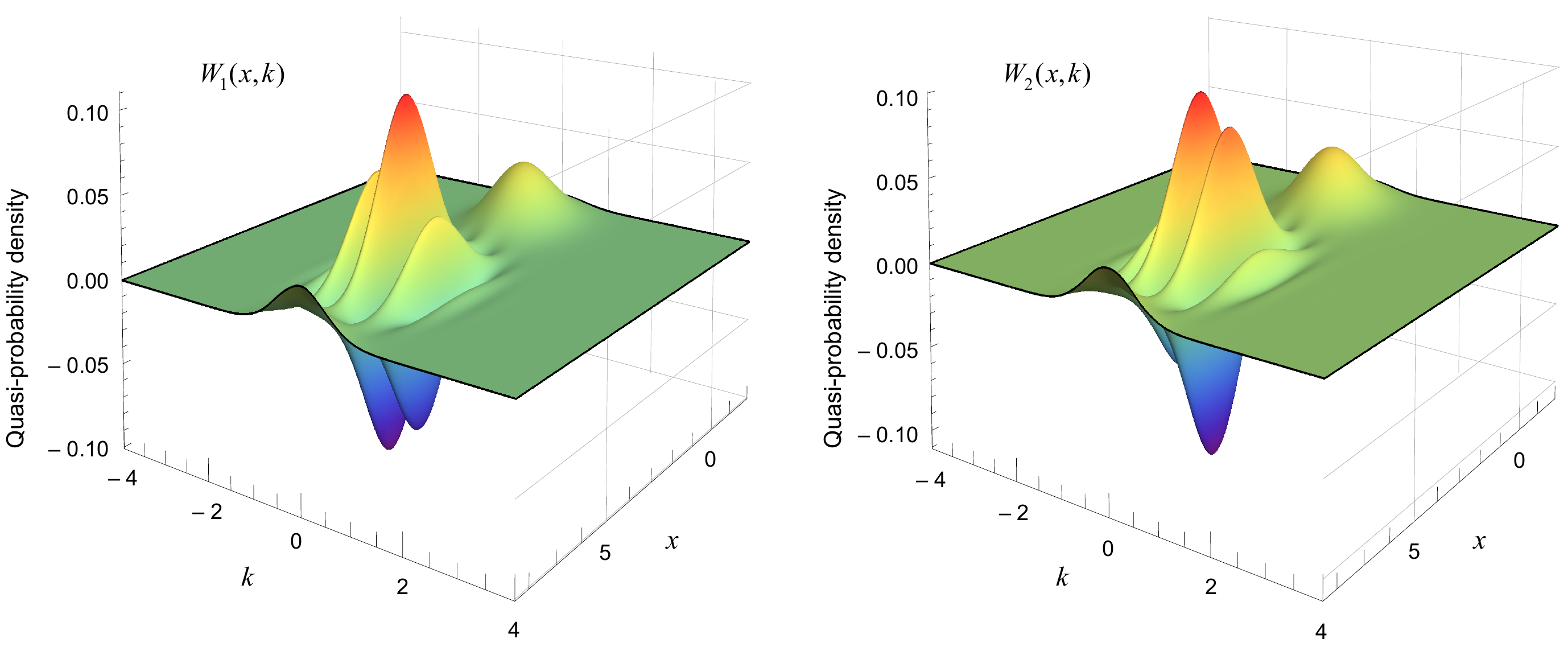}
\par\end{centering}
\caption{Plots of the Wigner functions $W_1(x,k)$ for $D_1$ (left) and $W_2(x,k)$ for $D_2$ (right) contain regions with negative quasi-probability density for $g=10$, $\sigma=1$ and $\varphi=\frac{\pi}{2}$.}
\label{fig:4}
\end{figure*}

Immediately before the second beam splitter $B_{2}$, the Wigner function of the single quantum particles is a probabilistic mixture
\begin{equation}
W_{+}(x,k)=\frac{1}{2}W_{A}(x,k)+\frac{1}{4}\left(1-\cos\varphi\right)W_{B}(x,k),
\end{equation}
of the non-negative Alice's and Bob's default Wigner functions
\begin{eqnarray}
W_{A}(x,k) & = & \frac{1}{\pi}e^{-\frac{\left(x-g\right)^{2}}{2\sigma^{2}}-2k^{2}\sigma^{2}},\\
W_{B}(x,k) & = & \frac{1}{\pi}e^{-\frac{x^{2}}{2\sigma^{2}}-2k^{2}\sigma^{2}}.
\end{eqnarray}

Application of the Wigner Uniqueness theorem identifies $f_{A}(x,k)=W_{A}(x,k)$ and $f_{B}(x,k)=W_{B}(x,k)$.

Similarly, after $B_{2}$, the Wigner Uniqueness theorem constrains the existence of the local hidden variables model through the requirement
\begin{align}
W_{i}(x,k)	=	\frac{1}{2}W_{A}(x,k)\,w_{i}^{A}\left(x,k,\varphi\right)
	 +\frac{1}{4}\left(1-\cos\varphi\right)W_{B}(x,k)\,w_{i}^{B}\left(x,k,\varphi\right) .
\end{align}
Unfortunately, the Wigner functions for the two detectors $D_{1}$ and $D_{2}$ are not non-negative
\begin{align}
W_{1}(x,k) = & \frac{1}{8\pi}e^{-2k^{2}\sigma^{2}}\left[e^{-\frac{x^{2}}{2\sigma^{2}}}{\left(1-\cos\varphi\right)}+2e^{-\frac{\left(x-g\right)^{2}}{2\sigma^{2}}}\right]\nonumber\\
&+\frac{1}{4\pi}e^{-2k^{2}\sigma^{2}}e^{-\frac{\left(2x-g\right)^{2}}{8\sigma^{2}}}\left[\cos\left(gk\right)-\cos\left(gk+\varphi\right)\right]\\
W_{2}(x,k) = & \frac{1}{8\pi}e^{-2k^{2}\sigma^{2}}\left[e^{-\frac{x^{2}}{2\sigma^{2}}}{\left(1-\cos\varphi\right)}+2e^{-\frac{\left(x-g\right)^{2}}{2\sigma^{2}}}\right]\nonumber\\
&-\frac{1}{4\pi}e^{-2k^{2}\sigma^{2}}e^{-\frac{\left(2x-g\right)^{2}}{8\sigma^{2}}}\left[\cos\left(gk\right)-\cos\left(gk+\varphi\right)\right]
\end{align}
The latter two Wigner functions exhibit negative regions for certain values of the setup parameters (see Fig.~\ref{fig:4}), which directly contradicts the two main assumptions of the local hidden variables model in Eqs.~\ref{eq:LHV-1} and~\ref{eq:LHV-2}, namely the requirement for valid probability densities, $f_{A}(x,k)\geq 0$, $f_{B}(x,k)\geq 0$, and valid weights for splitting, $0\leq w_{i}^{B}\left(x,k,\varphi\right)\leq 1$, $0\leq w_{i}^{A}\left(x,k,\varphi\right)\leq 1$.

The presented argument against the existence of a local hidden variables model illustrates the utility of Wigner function negativity as a witness of non-classical behavior \cite{Bertrand1987,Banaszek1999,Kenfack2004,Spekkens2008,Abramsky2011,Abramsky2014,Blass2015,Blass2018,Siyouri2016,Weinbub2018}. However, it should be noted that the uniqueness of the Wigner function depends on assumptions that may be rejected by the proponents of local hidden variable models:

On the one hand, giving Alice more choices for measurements of rotated quadratures appears to utilize the tomographic uniqueness of the Wigner function. However, this is helpful only if the hidden variables are somehow constrained to be just two, namely $x$ and $k$. If there is an infinite number of variables $\eta_{a,b}$, one for each linear superposition of position and momentum $\hat{\eta}_{a,b}=a\hat{x}+b\hat{k}$, then the hidden probability density functions will introduce infinitely multiple integrals.
Development of further mathematical arguments, for or against the existence of such an infinitely complex local hidden variables model, may be intractable.

On the other hand, Stenholm \cite{Stenholm1980} points out that the Wigner function is also uniquely defined if (i) it gives the correct marginal distributions for measurement of position and momentum, (ii) Galilean invariance and symmetry in position and momentum are valid, and (iii) for a free particle the classical equation of motion ensues. However, proponents of a local hidden variables model may object that at least one of the latter two conditions is violated by their model.

In essence, to determine whether the interferometric setup manifests quantum nonlocality it would be desirable to have a general argument for or against the simultaneous compatibility of Eqs.~\eqref{model-1}--\eqref{eq:LHV-2}.

Without changing the mathematical constraints on the local hidden variables model, we can rewrite Eqs.~\eqref{eq:LHV-1} and \eqref{eq:LHV-2} as follows:
\begin{eqnarray}
\Phi_{i}\left(x,\varphi\right) & = & \int_{-\infty}^{\infty}f_{+}(x,k,\varphi)\,w_{i}^{+}\left(x,k,\varphi\right)dk,\\
\Phi_{i}\left(k,\varphi\right) & = & \int_{-\infty}^{\infty}f_{+}(x,k,\varphi)\,w_{i}^{+}\left(x,k,\varphi\right)dx.
\end{eqnarray}
where the local hidden variables density before $B_2$ is
\begin{equation}
f_{+}(x,k,\varphi) = \frac{1}{2}f_{A}(x,k)+\frac{1}{4}\left(1-\cos\varphi\right)f_{B}(x,k)
\end{equation}
and the effective weight for splitting at $B_2$ is
\begin{equation}
w_{i}^{+}\left(x,k,\varphi\right) =  \frac{\frac{1}{2}f_{A}(x,k)}{f_{+}(x,k,\varphi)}\,w_{i}^{A}\left(x,k,\varphi\right)
 +\frac{\frac{1}{4}\left(1-\cos\varphi\right)f_{B}(x,k)}{f_{+}(x,k,\varphi)}w_{i}^{B}\left(x,k,\varphi\right)
\end{equation}
If Eqs.~\eqref{model-5}--\eqref{model-8} hold, it follows that $0\leq w_{i}^{+}\left(x,k,\varphi\right) \leq 1$.

In an attempt to construct explicit local hidden variables model, let us suppose that the hidden probability density function is factorizable
\begin{equation}
f_{+}(x,k,\varphi)=\frac{1}{\textrm{Prob}\left(D_{+}\right)}\Phi_{+}\left(x,\varphi\right)\Phi_{+}\left(k,\varphi\right)
\end{equation}
where $\textrm{Prob}\left(D_{+}\right)=\textrm{Prob}\left(D_{1}\right)+\textrm{Prob}\left(D_{2}\right)=\frac{1}{4}\left(3-\cos\varphi\right)$.
Then, the following system of integral equations arises
\begin{eqnarray}
w_{i}\left(x,\varphi\right) & = & \frac{1}{\textrm{Prob}\left(D_{+}\right)} \int_{-\infty}^{\infty}\Phi_{+}\left(k,\varphi\right)\,w_{i}^{+}\left(x,k,\varphi\right)dk, \label{factorize-1}\\
w_{i}\left(k,\varphi\right) & = & \frac{1}{\textrm{Prob}\left(D_{+}\right)} \int_{-\infty}^{\infty}\Phi_{+}\left(x,\varphi\right)\,w_{i}^{+}\left(x,k,\varphi\right)dx. \label{factorize-2}
\end{eqnarray}
In contrast to Fredholm theory, where the kernel is known and is solved for unknown function \cite{Kanwal1997,Wazwaz2011}, here we know the functions $w_{i}\left(x,\varphi\right)$, $w_{i}\left(k,\varphi\right)$, $\Phi_{+}\left(x,\varphi\right)$ and $\Phi_{+}\left(k,\varphi\right)$ (see \ref{app:distributions} with Eqs.~\eqref{Phi+x}, \eqref{Phi+k}, \eqref{w-x} and \eqref{w-k}) but the kernel is unknown because it is the hidden weight function $w_{i}^{+}\left(x,k,\varphi\right)$. What is needed is a mathematical method for solving this system of equations and checking that $0\leq w_{i}^{+}\left(x,k,\varphi\right) \leq 1$. Doing so may prove the existence of a local hidden variables model.

Assuming further factorizability of the hidden weight kernel for $D_{1}$ as $C_{1}w_{1}(x)w_{1}(k)$, followed by substituion in Eq.~\eqref{factorize-1} reveals the constant $C_1$. Thus, one solution of the system given by Eqs.~\eqref{factorize-1} and \eqref{factorize-2} is
\begin{equation}
w_{1}^{+}\left(x,k,\varphi\right) = \frac{\textrm{Prob}\left(D_{+}\right)}{\textrm{Prob}\left(D_{1}\right)}w_{1}(x)w_{1}(k);\quad
w_{2}^{+}\left(x,k,\varphi\right) = 1 - w_{1}^{+}\left(x,k,\varphi\right). \label{solution-1}
\end{equation}
Similarly, assuming factorizability of the hidden weight kernel for $D_{2}$ as $C_{2}w_{2}(x)w_{2}(k)$, reveals a second solution
\begin{equation}
w_{1}^{+}\left(x,k,\varphi\right) = 1 - w_{2}^{+}\left(x,k,\varphi\right); \quad
w_{2}^{+}\left(x,k,\varphi\right) =\frac{\textrm{Prob}\left(D_{+}\right)}{\textrm{Prob}\left(D_{2}\right)}w_{2}(x)w_{2}(k). \label{solution-2}
\end{equation}
Unfortunately, neither of the above two solutions is a valid local hidden variables model, because the weights in Eqs.~\eqref{solution-1} and \eqref{solution-2} may take on negative values or exceed 1 for certain choices of $\varphi$ by Bob.

The presented analysis illustrates that finding a solution such that $0\leq w_{i}^{+}\left(x,k,\varphi\right) \leq 1$, or conversely proving that no such solution exists in order to decide whether quantum weak values are nonlocal or not, may be a much more challenging problem than anticipated.

\section{Discussion}

Nonlocal hidden variables models experience no difficulties in reproducing the quantum outcomes from weak measurements. For example, the guiding equation of the actual (hidden) particle positions~$\mathbf{Q}_k$ for a quantum system with $N$ particles in de Broglie--Bohm model \cite{Bohm1,Bohm2,Allori2004,Durr2009}
is given by
\begin{equation}
\frac{d\mathbf{Q}_k}{dt} = \frac{\hbar}{m_k} \textrm{Im} \left[
{\frac{\psi^* \nabla_k \psi}{\psi^* \psi}}\right] (\mathbf{Q}_1 ,\ldots
,\mathbf{Q}_N)
\end{equation}
where the particle index $k$ runs from $1$ to $N$, $m_k$ is the mass of the $k$-th particle and $\nabla_k = (\frac{\partial}{\partial x_k},\frac{\partial}{\partial y_k},\frac{\partial}{\partial z_k})$ is the gradient with respect to the generic coordinates $\mathbf{q}_k = (x_k,y_k,z_k)$ of the $k$-th particle in the Schr\"{o}dinger picture ($\frac{d\mathbf{q}_k}{dt}=0$).
This guiding equation uses directly the nonlocal information contained in the many-body wavefunction $\psi(\mathbf{q}_1 ,\ldots ,\mathbf{q}_N)$, which solves the Schr\"{o}dinger equation
\begin{equation}
\imath\hbar\frac{\partial \psi}{\partial t} = \hat{H} \psi .
\end{equation}
Consequently, de Broglie--Bohm model is able to reproduce correctly the results obtained in the square nested Mach--Zehnder setup. This is only natural because the wavefunction $\psi$ from the Schr\"{o}dinger equation is already inbuilt as a nonlocal physical $\psi$-field \cite{Bohm1,Bohm2} into the formulation of de Broglie--Bohm model at a fundamental axiomatic level.

Alternatively, from the perspective of quantum information theory, there is no need for additional hidden variables. Instead, the quantum wavefunction $\psi$ is considered to be a complete, nonlocal description of reality. Furthermore, fundamental quantum no-go theorems establish that the quantum wavefunction $\psi$ of unknown quantum state cannot be cloned \cite{Wootters1982} and is not observable \cite{Busch1997}, which effectively makes $\psi$ itself analogous to a nonlocal hidden variable. If we possess multiple clones of a known quantum state, however, it is possible to use weak values extracted from weak measurements to perform quantum tomography in order to reconstruct the quantum state $\psi$ \cite{Wu2013,Kim2018}. Thus, the weak values appear to be robust physical properties of pre- and postselected quantum systems \cite{Vaidman2017} supporting the ontological interpretation of the quantum wave function $\psi$ \cite{Pusey2012,Leifer2014}.

In this work we have studied whether the weak values of quantum histories could be exploited to experimentally demonstrate nonlocality of single quanta. Based on the observation that the complex-valued weak values of quantum histories (Eq.~\ref{eq:weak}) correspond to relative quantum probability amplitudes that depend on all paths coherently explored in a quantum superposition by the quantum system \cite{Georgiev2018}, we have designed an interferometric setup in which the weak measuring device exhibits pointer distributions affected by both the real and imaginary part of the weak value (cf. Refs.~\refcite{Georgiev2018,Aharonov1988,Dressel2015,Jozsa2007,Aharonov2015,Aharonov2016b,Pusey2014,Cohen2017,Aharonov2018}).
Attempting to utilize the apparently nonlocal dependence of the weak values of quantum histories on spacelike separated choices, we have allowed Bob's choice $\varphi$ to be performed outside the light cone of the final projective measurement of $M$ by Alice. As expected, the generated quantum distributions by $M$ are functionally dependent on both Bob's and Alice's actions and this dependence needs to be replicated through local commitment of $M$ to particular $x$ or $k$ outcomes.

Importantly, the quantum mechanical predictions are the same regardless of how close to the second beam splitter $B_2$ the weak measuring device $M$ and the two detectors $D_1$ and $D_2$ are positioned.
The spatial arrangement of all these physical devices, however, is of great significance for the local hidden variables model.
If the distance from $B_2$ to $M$ is much shorter than the distance from $B_2$ to $D_1$ and $D_2$, then it would be possible for hidden signaling to travel backwards from $B_2$ to $M$ and overwrite any incorrect distribution for the pointer observable with the correct quantum distribution.
In contrast, if the distance from $B_2$ to $M$ is much larger than the distance from $B_2$ to $D_1$ and $D_2$, then it would be impossible for hidden signaling to travel at luminal speed backwards from $B_2$ and reach $M$ in time to overwrite an incorrect distribution for the pointer observable with the correct quantum distribution.
Thus, contextuality of local hidden variable models needs to operate in the absence of complete information about the spacelike separated actions by Alice and Bob. This imposes mathematical constraints on the local hidden variables model expressed by the system of equations~\eqref{model-1}--\eqref{eq:LHV-2}, which need to be obeyed for successful replication of the correct quantum distributions.

Our analysis based on certain techniques originally developed for use in phase space formulation of quantum mechanics showed that at the second beam splitter $B_2$ the task of the local hidden variables model is to split a quantum state characterized by a positive Wigner function into two quantum states with non-positive Wigner functions. This implies that under conditions for which the Wigner function is uniquely determined, such as those imposed by Blass and Gurevich \cite{Blass2015} or Stenholm \cite{Stenholm1980}, there exists no corresponding hidden variables model.
An alternative attempt to solve the system of equations~\eqref{model-1}--\eqref{eq:LHV-2} through postulation of simultaneous factorizability of both the hidden probability densities and weights of splitting, revealed two solutions, which were also not admissible as local hidden variables due to occurrence of negative weights of splitting.
Further progress on establishing the existence of a valid local hidden variables model of the given interferometric setup may require development of new mathematical techniques for solving integral equations with unknown nonfactorizable kernels.
While our study leaves that latter question as an open problem, we believe that the detailed quantum analysis of the presented experiment illustrates well the pitfalls encountered in relating quantum nonlocality with nonlocal interaction of the wavefunction $\psi$ of single quanta with spacelike separated devices.
Our application of Feynman's sum over histories approach also highlights the contextual nature of quantum weak values and hopefully provides the weak value community with a useful mathematical research tool.

\section*{Acknowledgements}

We wish to thank Avshalom Elitzur, Ebrahim Karimi and Lev Vaidman for helpful comments and discussions.

\appendix
\label{app}

\section{Nonlocal interpretation versus nonlocal ontology}
\label{app:nonlocal}

The quantum wavefunction $\psi$, which solves the Schr\"{o}dinger equation,
is subject to physical interpretation. In 1927, Einstein formulated
two different \emph{conceptions} of quantum theory that describe different
physical ontologies \cite{Einstein1927}. Later
work by John Bell showed that Einstein's two conceptions are not empirically
indistinguishable \emph{interpretations} of quantum theory, but are
empirically distinguishable \emph{models} that predict different
experimental outcomes for suitably designed experiments \cite{Bell1964,Bell1966}.

\emph{Conception 1}: The quantum wavefunction $\psi$ does not correspond
to a single particle, but to an \emph{ensemble} of particles extended
in space. The theory gives no information about individual processes,
but only about the ensemble of infinitely many elementary processes.
According to this statistical point of view, $|\psi|^{2}$ expresses
the probability that a particular particle from the ensemble exists
at the point considered, for example at a given point on a screen.
Thus, quantum theory is \emph{incomplete}, but can be completed by
\emph{local hidden variables} that localize individual particles during
their propagation.

\emph{Conception 2}: The quantum wavefunction $\psi$ corresponds
to \emph{single} particles and provides a \emph{complete} description
of individual processes. According to this ontological point of view,
$|\psi|^{2}$ expresses the probability that at a given instant the
same particle is present at a given point. Thus, quantum theory refers
to an individual process and claims to describe everything that is
governed by predefined laws. If $|\psi|^{2}$ is regarded as the probability
that a given particle is found at a certain point at a given time,
it could happen that the same elementary process produces an action
in two or several places on the screen, namely, two particles are
detected instead of one. Thus, to satisfy the conservation of energy
a \emph{nonlocal} mechanism of action at a distance is required, which
prevents the wave, continuously distributed in space, from producing
a simultaneous action in two places on the screen, namely, the detection of the
particle somewhere on the screen leads to nonlocal collapse of the
wavefunction manifested as nonlocal zeroing of the probability for
detecting the particle elsewhere.

Einstein rejected \emph{Conception 2} because for him it contradicted the
theory of general relativity. He believed that quantum nonlocality
is an interpretation-dependent artifact that is a consequence of the
assumed completeness of quantum theory. Following the work of Bell,
however, now we know that \emph{Conception 1} endorsing local hidden
variables is not an interpretation of quantum theory, but a physical
model that cannot reproduce all quantum mechanical experiments.

In this work, we studied whether quantum physics involves nonlocal
ontology for single particles that is not an artifact of the nonlocal
interpretation of the quantum wavefunction $\psi$.
This required usage of a general local hidden variables model that
is able to replicate the experimental outcomes of Einstein's experiment
and Wheeler's delayed choice experiment. Next, we show that the main
characteristics of the general local hidden variables model are
constrained by the requirement to replicate the experimental outcomes
of the latter two experiments.

\section{Local hidden variables model of Einstein's experiment}
\label{app:Einstein}

\begin{figure}[t]
\begin{centering}
\includegraphics[width=80mm]{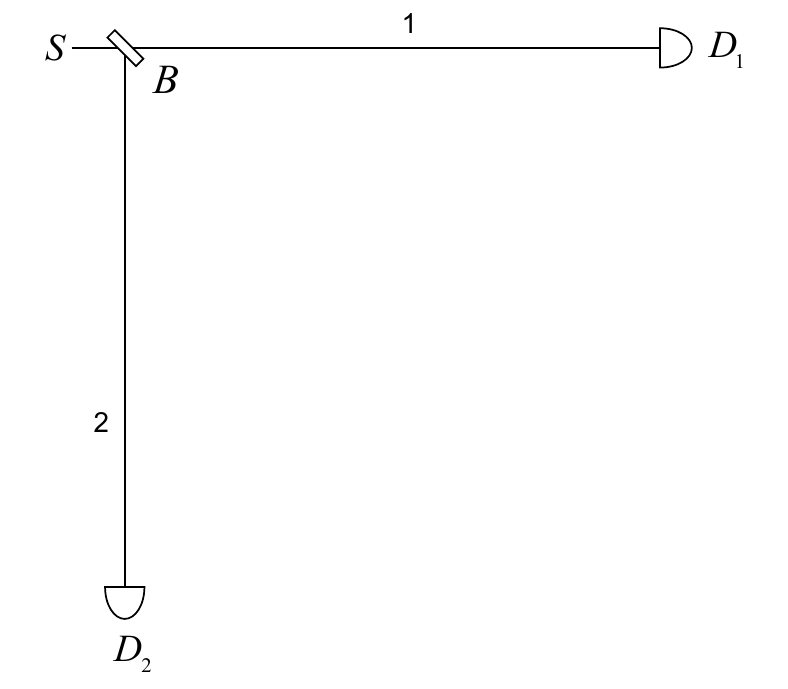}
\par\end{centering}
\caption{\label{fig:Einstein}Einstein's experiment constructed with the use
of a single-photon source $S$, a beam splitter $B$ and two distant detectors
$D_{1}$ and $D_{2}$, which are never found to click together at
the same time as required by the conservation of energy.}
\end{figure}

The quantum mechanical description of Einstein's experiment (Fig.~\ref{fig:Einstein}) involves
unitary evolution $\hat{\mathcal{T}}_{f,i}$ of the initial quantum
state $|S\rangle$ into a quantum superposition of two distant final
states $|D_{1}\rangle$ and $|D_{2}\rangle$ as follows
\begin{equation}
\hat{\mathcal{T}}_{f,i}|S\rangle=\frac{1}{\sqrt{2}}\left(|D_{1}\rangle+\imath|D_{2}\rangle\right)
\end{equation}
Upon measurement in position basis at time $t_{f}$, the final superposed
state collapses in one of the two final states $|D_{1}\rangle$ and
$|D_{2}\rangle$ with probabilities given by the Born rule
\begin{equation}
\begin{cases}
\frac{1}{\sqrt{2}}\left(|D_{1}\rangle+\imath|D_{2}\rangle\right)\to|D_{1}\rangle & \textrm{with Prob}(D_{1})=\frac{1}{2}\\
\frac{1}{\sqrt{2}}\left(|D_{1}\rangle+\imath|D_{2}\rangle\right)\to|D_{2}\rangle & \textrm{with Prob}(D_{2})=\frac{1}{2}
\end{cases}
\end{equation}
The wavefunction collapse at $t_{f}$ appears to be a nonlocal event,
because the conservation of energy requires the probabilities of firing
of the two distant detectors $D_{1}$ and $D_{2}$ to be dependent,
namely, $D_{1}$ never clicks if $D_{2}$ does, and vice versa. In
fact, accepting the quantum mechanical description of individual quantum
processes as complete would make all quantum measurements to appear
nonlocal at the time of measurement.

As noticed by Einstein, however, it is possible to construct a local
hidden variables model by shifting the wavefunction collapse to an
earlier time $t_{b}$ when the quantum particle passes through the beam splitter
$B$. Einstein's strategy is to explain the observed correlations
between distant measurement outcomes through past local selections
at points of bifurcation of quantum trajectories. The passage through
the beam splitter is a local event, hence the quantum particle should
be able to perform locally a \emph{weighted random choice} $\mathcal{R}$
with probabilities given by a statistical distribution $\Lambda$
that is consistent with the predictions of quantum mechanics. Thus,
local ontology could be restored in Einstein's experiment only if
the local hidden variables model is endowed with the following characteristic
property:

\emph{Property 1}: Quantum particles possess a probabilistic mechanism
that performs weighted random choice $\mathcal{R}$ with
probabilities given by any statistical distribution $\Lambda$.
This would allow quantum particles to select a single path at points
of bifurcation of quantum trajectories.

For the present purposes, we could grant the existence of a genuinely
indeterministic random choice $\mathcal{R}$, which is due to a true
Random Number Generator (RNG) as opposed to deterministic Pseudo-Random
Number Generator (PRNG). Thus, the general local hidden
variables model considered here could possess an additional resource
that is unavailable to deterministic classical physical models.

\section{Local hidden variables model of Wheeler's delayed choice experiment}
\label{app:Wheeler}

\begin{figure}[t]
\begin{centering}
\includegraphics[width=80mm]{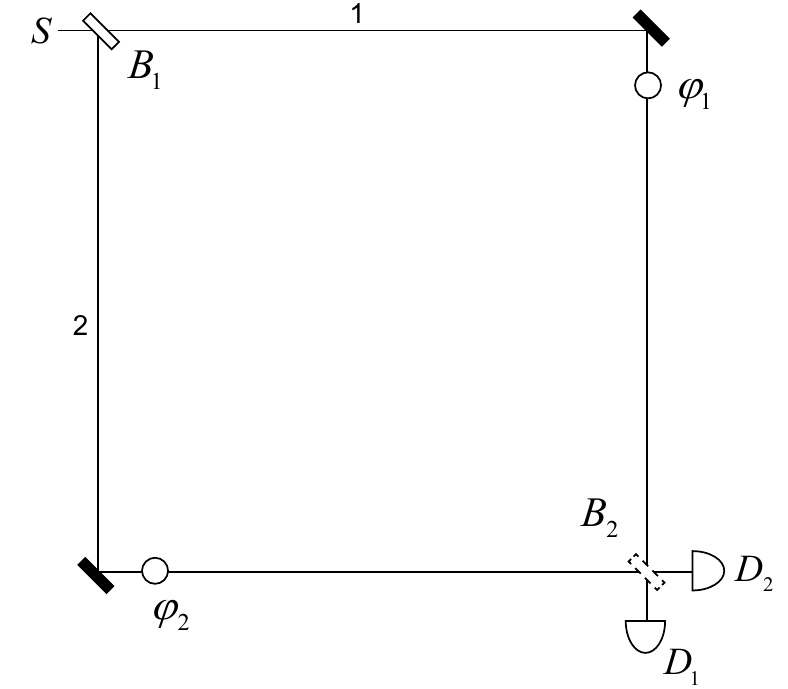}
\par\end{centering}
\caption{\label{fig:Wheeler}Wheeler's delayed choice experiment performed with the use of a Mach--Zehnder interferometer.
The presence of the second beam splitter~$B_{2}$ is decided
in a delayed fashion only after the particle has passed the first
beam splitter~$B_{1}$.}
\end{figure}

Quantum interference effects are sensitive to all physical influences
exerted on the available alternative quantum trajectories for the
particle. The quantum mechanical description of Wheeler's delayed
choice experiment (Fig.~\ref{fig:Wheeler}) involves unitary evolution $\hat{\mathcal{T}}_{f,i}$
of the initial quantum state $|S\rangle$ into a quantum superposition
of two alternative trajectories $|1\rangle$ and $|2\rangle$ that
may or may not be recombined by a second beam splitter $B_{2}$ whose
insertion is decided in a delayed fashion only after the particle
has passed the first beam splitter $B_{1}$.

In the absence of the second beam splitter $B_{2}$, the unitary
evolution $\hat{\mathcal{T}}_{f,i}$ of the initial quantum state
$|S\rangle$ leads to a quantum superposition of two final detector
states $|D_{1}\rangle$ and $|D_{2}\rangle$ as follows
\begin{equation}
\hat{\mathcal{T}}_{f,i}|S\rangle=\frac{1}{\sqrt{2}}\left(\imath e^{\imath\varphi_{1}}|D_{1}\rangle-e^{\imath\varphi_{2}}|D_{2}\rangle\right)
\end{equation}
The particle arrives at each of the two detectors with equal probability of $\frac{1}{2}$. Because the quantum amplitude reaching detector
$D_{1}$ or $D_{2}$ picks the extra phase by only one of the phase shifters $\varphi_{1}$ or $\varphi_{2}$, it appears that the
particle has traveled through the corresponding interferometer arm~$1$ or $2$.

In the presence of the second beam splitter $B_{2}$, however,
the unitary evolution $\hat{\mathcal{T}}_{f,i}$ of the initial quantum
state $|S\rangle$ leads to quantum interference of the phase shifts
picked up from both interferometer arms
\begin{equation}
\hat{\mathcal{T}}_{f,i}|S\rangle=\frac{1}{2}\left[\imath\left(e^{\imath\varphi_{1}}-e^{\imath\varphi_{2}}\right)|D_{1}\rangle-\left(e^{\imath\varphi_{1}}+e^{\imath\varphi_{2}}\right)|D_{2}\rangle\right]
\end{equation}
The particle arrives at $D_{1}$ with probability of $\frac{1-\cos(\varphi_{1}-\varphi_{2})}{2}$
and at $D_{2}$ with probability of $\frac{1+\cos(\varphi_{1}-\varphi_{2})}{2}$.
Clearly, the particle needs to have the information about
the values of both phase shifters $\varphi_{1}$ and $\varphi_{2}$
in order to produce the observable quantum outcomes. Because according
to \emph{Property 1} the particle can take probabilistically only
a single trajectory, the introduction of hidden signaling is necessary
in order to reproduce the experimental outcomes in Wheeler's delayed
choice experiment. Thus, the local hidden variables model should also
be endowed with the following characteristic properties:

\emph{Property 2}: Quantum particles could explore available alternative
trajectories with the use of hidden signals that travel at most at
luminal speed.

\emph{Property 3}: Quantum particles posses memory and could execute
a list of contextual instructions, which allow usage of new information
obtained through hidden signals.

For the present purposes, we could grant the existence of any type
of hidden signals (particles or waves) provided that the hidden signals cannot be directly observed and do not generate physical particles upon quantum measurement (so called ``empty waves''
in de Broglie--Bohm model). Furthermore, we could grant an arbitrarily
large memory and processing power available to the quantum particle
for making use of the hidden signals. The important constraints are
that the velocity of the hidden signals is not superluminal and distant
changes in the physical system do not have an instantaneous effect
upon the current list of contextual instructions executed by the quantum
particle.

With the above properties, the general local hidden variables model
easily replicates the measurement outcomes in Wheeler's delayed choice
experiment: The quantum particle probabilistically takes a single
path in the interferometer and picks up locally the phase shift encountered, $\varphi_{1}$
or $\varphi_{2}$. When the particle arrives at $B_{2}$,
it receives a hidden signal about the phase shift on the alternative
path, $\varphi_{2}$ or $\varphi_{1}$. Lastly, if $B_2$ is absent, the particle goes to $D_1$ if it traveled along path~$1$ and to $D_2$ if it traveled along path~$2$. Else, if $B_2$ is present, the particle uses its memory and processing power to
compute the quantum interference patterns $\textrm{Prob}(D_1)=\frac{1-\cos(\varphi_{1}-\varphi_{2})}{2}$ and $\textrm{Prob}(D_2)=\frac{1+\cos(\varphi_{1}-\varphi_{2})}{2}$, after which with the use of its weighted random choice mechanism $\mathcal{R}$ selects
one of the two detectors $D_{1}$ or $D_{2}$ in accordance with the
computed quantum probabilities.

\section{Analytic derivation of the final meter wavefunction}
\label{app:analytic}

To solve for the exact final meter state analytically, we use the formal definition of the matrix exponential of the interaction Hamiltonian as an infinite power series \cite{Kofman2012,Koike2011,Nakamura2012}. Once the complex-valued weak value $A_w$ is generated using inner products of the projection operator $\hat{A}$, the obtained infinite power series could be transformed into a sum of two exponentials. This mathematical procedure elucidates the origin of the weak values and is well-defined for any strength of the coupling parameter~$g$.

For the post-selected system in a final state $|\psi_{f}\rangle$,
the projected (not normalized) final meter wavefunction in position
basis can be expanded as follows
\begin{align}
\langle x|\phi_{f}\rangle = & \langle\psi_{f}|\hat{\mathcal{T}}_{f,m}\left(e^{-\imath g\hat{A}\otimes\hat{k}}\right)\hat{\mathcal{T}}_{m,i}|\psi_{i}\rangle\phi_{0}(x)\nonumber \\
 = & \langle\psi_{f}|\hat{\mathcal{T}}_{f,m}\left[\sum_{n=0}^{\infty}\frac{1}{n!}\left(-\imath g\hat{A}\otimes\hat{k}\right)^{n}\right]\hat{\mathcal{T}}_{m,i}|\psi_{i}\rangle\phi_{0}(x)\nonumber \\
 = & \langle\psi_{f}|\hat{\mathcal{T}}_{f,m}\left(1-\imath g\hat{A}\otimes\hat{k}-\frac{1}{2!}g^{2}\hat{A}^{2}\otimes\hat{k}^{2}+\frac{1}{3!}\imath g^{3}\hat{A}^{3}\otimes\hat{k}^{3}+\ldots\right)\hat{\mathcal{T}}_{m,i}|\psi_{i}\rangle\phi_{0}(x).
\end{align}
For the particular case of a projection operator represented by an idempotent matrix, $\hat{A}^{2}=\hat{A}$,
after expressing the wavenumber operator in position basis $\hat{k}=-\imath\frac{\partial}{\partial x}$, we obtain
\begin{equation}
\phi_{f}(x)=\langle\psi_{f}|\hat{\mathcal{T}}_{f,i}|\psi_{i}\rangle\left[1-A_{w}\left(g\frac{\partial}{\partial x}-\frac{1}{2!}\left(g\frac{\partial}{\partial x}\right)^{2}+\frac{1}{3!}\left(g\frac{\partial}{\partial x}\right)^{3}-\ldots\right)\right]\phi_{0}(x)
\end{equation}
where $A_{w}$ is the weak value defined in Eq.~\eqref{eq:weak}. Then we employ the fact that the exponential of the differential operator acts as a translation operator
\begin{align}
\phi_{f}(x) = & \langle\psi_{f}|\hat{\mathcal{T}}_{f,i}|\psi_{i}\rangle\left\{ \left(1-A_{w}\right)+A_{w}\left[1-g\frac{\partial}{\partial x}+\frac{1}{2!}\left(g\frac{\partial}{\partial x}\right)^{2}-\frac{1}{3!}\left(g\frac{\partial}{\partial x}\right)^{3}+\ldots\right]\right\} \phi_{0}(x)\nonumber \\
  = & \langle\psi_{f}|\hat{\mathcal{T}}_{f,i}|\psi_{i}\rangle\left[\left(1-A_{w}\right)+A_{w}\sum_{n=0}^{\infty}\frac{1}{n!}\left(-g\frac{\partial}{\partial x}\right)^{n}\right]\phi_{0}(x)\nonumber \\
  = & \langle\psi_{f}|\hat{\mathcal{T}}_{f,i}|\psi_{i}\rangle\left[\left(1-A_{w}\right)+A_{w}e^{-g\frac{\partial}{\partial x}}\right]\phi_{0}(x)\nonumber\\
	= & \langle\psi_{f}|\hat{\mathcal{T}}_{f,i}|\psi_{i}\rangle\left[\left(1-A_{w}\right)\phi_{0}(x)+A_{w}\phi_{0}(x-g)\right]\nonumber \\
  = & \langle\psi_{f}|\hat{\mathcal{T}}_{f,i}|\psi_{i}\rangle\left(2\pi\sigma^{2}\right)^{-\frac{1}{4}}\left[\left(1-A_{w}\right)e^{-\frac{x^{2}}{4\sigma^{2}}}+A_{w}e^{-\frac{(x-g)^{2}}{4\sigma^{2}}}\right]
\end{align}
The presented mathematical technique is also able to generate exact analytical results for observables represented by an involutory matrix, $\hat{A}^{2}=\hat{I}$, however, it may not lead to simple expressions in a closed form for an arbitrary observable $\hat{A}$ \cite{Kofman2012,Koike2011,Nakamura2012}.

\section{Quantum distributions of the pointer for different post-selections}
\label{app:distributions}

The correct quantum distributions for post-selected detectors $D_1$ or $D_2$ that need to be reproduced by the local hidden variables model can be obtained from Eqs.~\eqref{eq:Phi-f-x} and~\eqref{eq:Phi-f-k} with the use of the corresponding weak values in Eqs.~\eqref{eq:Aw-D1} and~\eqref{eq:Aw-D2} as follows
\begin{align}
\Phi_{1}(x)  = & \frac{1}{8}\left(2\pi\sigma^{2}\right)^{-\frac{1}{2}}e^{-\frac{x^{2}}{2\sigma^{2}}}\left[2e^{\frac{x^{2}-\left(x-g\right)^{2}}{2\sigma^{2}}}+\left(1+2e^{\frac{x^{2}-\left(x-g\right)^{2}}{4\sigma^{2}}}\right)\left(1-\cos\varphi\right)\right],\label{Phi-1-x}\\
\Phi_{2}(x)  = & \frac{1}{8}\left(2\pi\sigma^{2}\right)^{-\frac{1}{2}}e^{-\frac{x^{2}}{2\sigma^{2}}}\left[2e^{\frac{x^{2}-\left(x-g\right)^{2}}{2\sigma^{2}}}+\left(1-2e^{\frac{x^{2}-\left(x-g\right)^{2}}{4\sigma^{2}}}\right)\left(1-\cos\varphi\right)\right],\label{Phi-2-x}\\
\Phi_{1}(k)  = & \frac{1}{8}\sigma\sqrt{\frac{2}{\pi}}e^{-2k^{2}\sigma^{2}}\left[3-\cos\varphi+2\cos\left(gk\right)-2\cos\left(\varphi+gk\right)\right],\label{Phi-1-k}\\
\Phi_{2}(k)  = & \frac{1}{8}\sigma\sqrt{\frac{2}{\pi}}e^{-2k^{2}\sigma^{2}}\left[3-\cos\varphi-2\cos\left(gk\right)+2\cos\left(\varphi+gk\right)\right],\label{Phi-2-k}\\
\Phi_{1}(\eta)  = & \frac{\sigma}{4\sqrt{2\pi(b^{2}+4a^{2}\sigma^{4})}}\left\{ 2e^{-\frac{2(\eta-ag)^{2}\sigma^{2}}{b^{2}+4a^{2}\sigma^{4}}}+e^{-\frac{2\eta^{2}\sigma^{2}}{b^{2}+4a^{2}\sigma^{4}}}(1-\cos\varphi)\right\}\nonumber\\
&+\frac{\sigma}{\sqrt{8\pi(b^{2}+4a^{2}\sigma^{4})}}\left\{e^{-\frac{[\eta^{2}+(\eta-ag)^{2}]\sigma^{2}}{b^{2}+4a^{2}\sigma^{4}}}\left[\cos{\scriptstyle \left(\frac{bg(2\eta-ag)}{2(b^{2}+4a^{2}\sigma^{4})}\right)}-\cos{\scriptstyle \left(\frac{bg(2\eta-ag)}{2(b^{2}+4a^{2}\sigma^{4})}+\varphi\right)}\right]\right\}, \label{Phi-1-eta}\\
\Phi_{2}(\eta)  = & \frac{\sigma}{4\sqrt{2\pi(b^{2}+4a^{2}\sigma^{4})}}\left\{ 2e^{-\frac{2(\eta-ag)^{2}\sigma^{2}}{b^{2}+4a^{2}\sigma^{4}}}+e^{-\frac{2\eta^{2}\sigma^{2}}{b^{2}+4a^{2}\sigma^{4}}}(1-\cos\varphi)\right\}\nonumber\\
&-\frac{\sigma}{\sqrt{8\pi(b^{2}+4a^{2}\sigma^{4})}}\left\{e^{-\frac{[\eta^{2}+(\eta-ag)^{2}]\sigma^{2}}{b^{2}+4a^{2}\sigma^{4}}}\left[\cos{\scriptstyle \left(\frac{bg(2\eta-ag)}{2(b^{2}+4a^{2}\sigma^{4})}\right)}-\cos{\scriptstyle \left(\frac{bg(2\eta-ag)}{2(b^{2}+4a^{2}\sigma^{4})}+\varphi\right)}\right]\right\}. \label{Phi-2-eta}
\end{align}

The probabilities of detector clicking for $D_{1}$ or $D_{2}$ in the presence of the
weak measuring device are
\begin{align}
\textrm{Prob}\left(D_{1}\right) = &~ \frac{1}{8}\left[3-\cos\varphi+2e^{-\frac{g^{2}}{8\sigma^{2}}}\left(1-\cos\varphi\right)\right], \label{eq:Prob-D1}\\
\textrm{Prob}\left(D_{2}\right) = &~ \frac{1}{8}\left[3-\cos\varphi-2e^{-\frac{g^{2}}{8\sigma^{2}}}\left(1-\cos\varphi\right)\right]. \label{eq:Prob-D2}
\end{align}

For post-selected $D_f$, the normalized final meter distribution in basis $j$ is
\begin{equation}
\tilde{\Phi}_{f} (j)=\frac{1}{\textrm{Prob}(D_{f})}\Phi_{f} (j).
\end{equation}

\section{Mixing of hidden position and momentum outcomes}

\begin{figure*}[t]
\begin{centering}
\includegraphics[width=125mm]{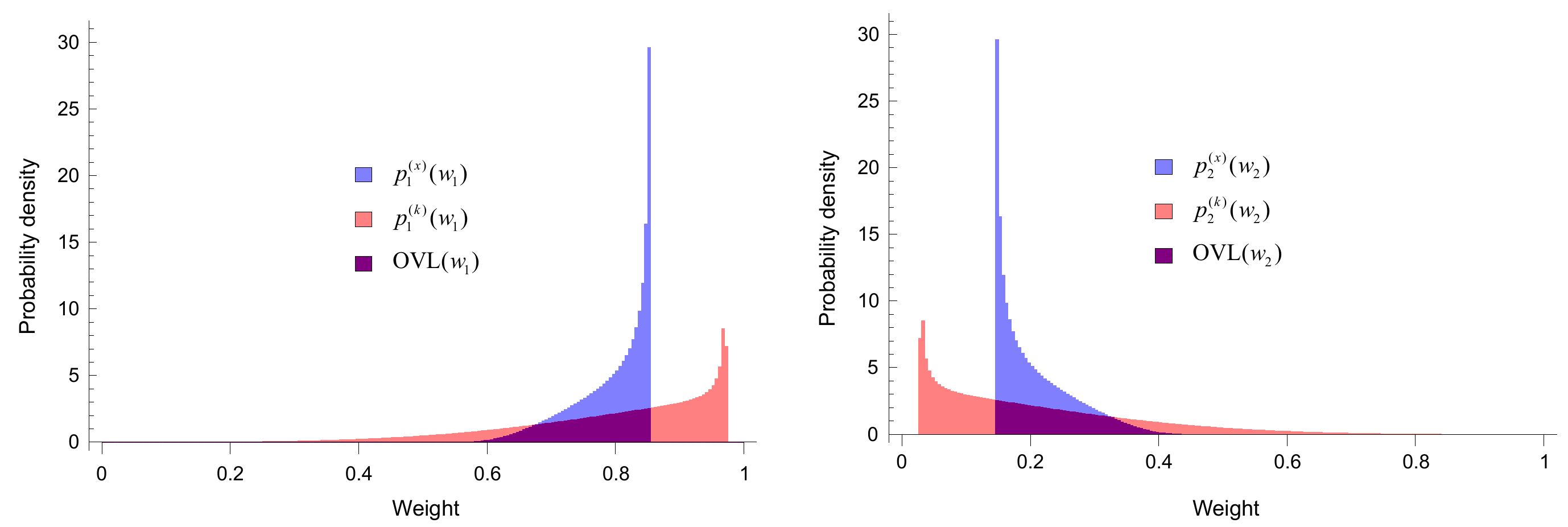}
\par\end{centering}
\caption{Plots of weight--probability histograms of $\Phi_{+}(x)$ and $\Phi_{+}(k)$ splitting towards $D_1$ or $D_2$ for $g=1$, $\sigma=1$ and $\varphi=\frac{\pi}{2}$.}
\label{fig:7}
\end{figure*}

\begin{definition}
(Weight--probability histogram of a composite distribution) Let $\Phi_{+}(\lambda)$ be a normalized probability distribution composed of two other distributions $\Phi_{+}(\lambda) = \Phi_{1}(\lambda) + \Phi_{2}(\lambda)$. The weight functions for splitting
\begin{equation}
w_{1}(\lambda)=\frac{\Phi_{1}(\lambda)}{\Phi_{+}(\lambda)},\quad\quad w_{2}(\lambda)=\frac{\Phi_{2}(\lambda)}{\Phi_{+}(\lambda)},
\end{equation}
are bound within the range $[0,1]$ and \mbox{$w_{1}(\lambda)+w_{2}(\lambda)=1$}.
Divide the weight function range into $n$ bins of width~$\Delta w$.
The probability $p_i(n)$ of the $n$th bin in the histogram is given by integration of $\Phi_{+}(\lambda)$ over the domain region $R_{\lambda}$ for which $(n-1)\Delta w \leq w_i(\lambda) \leq n \Delta w$. Compactly, the histogram is defined by the following integral measure
\begin{equation}
p_{i}^{(\lambda)}(n)  =  \int_{-\infty}^{\infty}\left\{\theta\left[w_{i}\left(\lambda\right)-(n-1)\Delta w\right]-\theta\left[w_{i}\left(\lambda\right)-n\Delta w\right]\right\}\Phi_{+}(\lambda)\,d\lambda
\end{equation}
where $\theta(\cdot)$ is the Heaviside step function. To convert the discrete probability histogram into a continuous probability density histogram, one needs to take the limit
\begin{equation}
p_{i}^{(\lambda)}(w_{i})=  \lim_{\Delta w\to0}\int_{-\infty}^{\infty}\frac{\theta\left[w_{i}\left(\lambda\right)-(n-1)\Delta w\right]-\theta\left[w_{i}\left(\lambda\right)-n\Delta w\right]}{\Delta w}\Phi_{+}(\lambda)\,d\lambda
\end{equation}
Then, the probability is recovered through integration with respect to the weight in the probability density histogram
\begin{equation}
\int_{0}^{1}p_{i}^{(\lambda)}(w_{i})dw_{i}=\int_{-\infty}^{\infty}\Phi_{i}(\lambda)\,d\lambda \label{eq:histo}
\end{equation}
Since $w_{1}(\lambda)+w_{2}(\lambda)=1$, the histograms for each of the two component distributions are with identical shape but reversed bins, $p_{1}^{(\lambda)}(w)=p_{2}^{(\lambda)}(1-w)$.
\end{definition}

\setcounter{theorem}{2}
\begin{theorem}
\label{theorem3}
Let $\Phi_{+}(x)$ and $\Phi_{+}(k)$ be two incompatible probability density distributions only one of which could actualize a measurement outcome, $x_j$ or $k_j$, per each run~$j$.
Without knowledge of which distribution will be used to generate the actual outcome for each run $j$, the two distributions could be split into component distributions such that $\Phi_{+}(x) = \Phi_{1}(x) + \Phi_{2}(x)$ and $\Phi_{+}(k) = \Phi_{1}(k) + \Phi_{2}(k)$, provided that
\begin{equation}
\int_{-\infty}^\infty \Phi_{i}(x) dx = \int_{-\infty}^\infty \Phi_{i}(k) dk, \label{eq:cond-1}
\end{equation}
even though the weight--probability histograms of $\Phi_{+}(x)$ and $\Phi_{+}(k)$ may not be identical.
\end{theorem}
\begin{proof}
Suppose that one of the weight--probability density histograms is wider than the other, e.g. $p_i^{(k)}(w_i)$ is wider than $p_i^{(x)}(w_i)$.
The overlap region is
\begin{equation}
\textrm{OVL}(w_{i})=\frac{p_{i}^{(k)}(w_{i})+p_{i}^{(x)}(w_{i})-\left|p_{i}^{(k)}(w_{i})-p_{i}^{(x)}(w_{i})\right|}{2}.
\end{equation}
Because the overlap region is identical for the two histograms, outcomes for $x$ or $k$ in $\textrm{OVL}(w_{i})$ could be directly matched based on their weights of splitting. What is left to be matched are the two histogram remainders,
\mbox{$p_{i}^{(x)}(w_{i}) - \textrm{OVL}(w_{i})$} and
\mbox{$p_{i}^{(k)}(w_{i}) - \textrm{OVL}(w_{i})$}.
Since $p_{i}^{(k)}(w_{i})$ was assumed to be wider than $p_{i}^{(x)}(w_{i})$, it follows that
\mbox{$p_{i}^{(k)}(w_{i}) - \textrm{OVL}(w_{i})$} will contain bins with weights outside the range of
\mbox{$p_{i}^{(x)}(w_{i}) - \textrm{OVL}(w_{i})$}.
Then from Eqs.~\ref{eq:histo} and \ref{eq:cond-1}, it follows that in the histogram remainders the total probabilities of splitting towards each of the detectors are the same.
Consequently, the weights for $x$ outcomes in \mbox{$p_{i}^{(x)}(w_{i}) - \textrm{OVL}(w_{i})$} could indeed be generated as averages from mixing of $k$ outcomes with lower and higher weights in \mbox{$p_{i}^{(k)}(w_{i}) - \textrm{OVL}(w_{i})$}.
\end{proof}

The parameter values for the interferometric setup, including Bob's choice of $\varphi$, could be chosen so that the the weight--probability histograms of $\Phi_{+}(x)$ and $\Phi_{+}(k)$ are manifestly different (Fig.~\ref{fig:7}).
For $g=1$, $\sigma=1$ and $\varphi=\frac{\pi}{2}$,
the quantum weights for splitting have largely different upper
and lower bounds, namely,
\mbox{$0.5\leq w_{1}(x)\leq0.854$} vs.
\mbox{$0.029\leq w_{1}(k)\leq0.971$}, and
\mbox{$0.146\leq w_{2}(x)\leq0.5$} vs.
\mbox{$0.029\leq w_{2}(k)\leq0.971$} (cf. Fig.~\ref{fig:3}).
Theorem~\ref{theorem3} does not directly help the construction of local hidden variables model, however, because the histograms of $\Phi_{+}(x)$ and $\Phi_{+}(k)$ depend strongly on Bob's choice~$\varphi$ and this information is not available at the time of commitment to an outcome by Alice's measuring device~$M$ when the mixing of $x$ and $k$ outcomes has to be done.
Instead, the local hidden variables model should use information for $\varphi$ only at $B_2$ as shown in Eqs.~\eqref{eq:LHV-1} and \eqref{eq:LHV-2}.

\end{document}